\documentclass[final,hidelinks,onefignum,onetabnum]{siamart251216}

\usepackage{amsmath,amssymb}
\usepackage{mathtools}
\usepackage{bm}
\usepackage{graphicx}
\usepackage{booktabs}
\usepackage{enumitem}
\usepackage{algorithm}
\usepackage{algpseudocode}

\newcommand{\RR}{\mathbb{R}}
\newcommand{\NN}{\mathbb{N}}
\newcommand{\EE}{\mathbb{E}}

\newcommand{\Var}{\mathrm{Var}}
\newcommand{\Cov}{\mathrm{Cov}}
\newcommand{\DKL}{\mathcal{D}_{\mathrm{KL}}}
\newcommand{\tr}{\operatorname{Tr}}

\newcommand{\gauss}{\gamma}

\newcommand{\mA}{\bm{A}}
\newcommand{\mb}{\bm{b}}
\newcommand{\mC}{\bm{C}}
\newcommand{\mB}{\bm{B}}
\newcommand{\hmC}{\widehat{\bm{C}}}
\newcommand{\mV}{\bm{V}}
\newcommand{\hmV}{\widehat{\bm{V}}}

\newcommand{\mP}{\bm{P}}
\newcommand{\hmP}{\widehat{\bm{P}}}

\newcommand{\mQ}{\bm{Q}}

\newcommand{\mI}{\bm{I}}

\newcommand{\mR}{\bm{R}}
\newcommand{\mD}{\bm{D}}

\newcommand{\mH}{\bm{H}}

\newcommand{\mM}{\bm{M}}

\newcommand{\St}{\mathrm{St}}
\newcommand{\Gr}{\mathrm{Gr}}
\newcommand{\whvtheta}{\widehat{\bm{\theta}}^{d}}

\newcommand{\hvbeta}{\widehat{\bm{\theta}}^{r}}

\newcommand{\vn}{\bm{n}}

\newcommand{\vx}{\bm{x}}

\newcommand{\vz}{\bm{z}}
\newcommand{\vZ}{\bm{Z}}
\newcommand{\vu}{\bm{u}}
\newcommand{\vU}{\bm{U}}
\newcommand{\vv}{\bm{v}}
\newcommand{\vw}{\bm{w}}
\newcommand{\vW}{\bm{W}}

\newcommand{\scZ}{\mathcal{Z}}
\newcommand{\scL}{\mathcal{L}}
\newcommand{\scS}{\mathcal{S}}
\newcommand{\scF}{\mathcal{F}}

\newsiamremark{remark}{Remark}
\newsiamremark{assumption}{Assumption}
\crefname{remark}{Remark}{Remarks}
\crefname{assumption}{Assumption}{Assumptions}
\let\origremarkenv\remark
\renewcommand{\remark}{\crefalias{theorem}{remark}\origremarkenv}
\let\origassumptionenv\assumption
\renewcommand{\assumption}{\crefalias{theorem}{assumption}\origassumptionenv}
\crefname{appendix}{Appendix}{Appendices}
\Crefname{appendix}{Appendix}{Appendices}

\headers{Copula Active Subspaces}{J.~Chen and P.~J.~van Leeuwen}

\title{Copula Active Subspaces II: Error Decomposition, A Posteriori Estimation, and Sharpness of the Bounds\thanks{This is Part~II of a two-part work. The companion paper \emph{Copula Active Subspaces I: A Score-Covariance Method for Reduced-Order Non-Gaussian Density Estimation} defines the estimator analyzed here. \funding{This work was supported by the Consortium for Advanced Data Assimilation Research and Education (CADRE), funded by NOAA under grant 2007893.}}}

\author{%
  Joshua Chen\thanks{Department of Atmospheric Science, Colorado State University, Fort Collins, CO 80523 (\email{Joshua.Chen@colostate.edu}).}
  \and
  Peter Jan van Leeuwen\thanks{Department of Atmospheric Science, Colorado State University, Fort Collins, CO 80523 (\email{Peter.vanLeeuwen@colostate.edu}).}
}

\newcommand{\vzero}{\bm{0}}

\newcommand{\hmA}{\widehat{\bm{A}}}

\newcommand{\PartIref}[2]{\csname Cref@#1@name\endcsname~#2}
\hypersetup{hypertexnames=false}
\let\CASsection\section

\begin{document}
\maketitle

\begin{abstract}

The Copula Active Subspaces (\textsc{Cas}) method of Part~I estimates a non-Gaussian observation-noise density $\pi_{\vn}$ from residual samples. Its KL divergence $\DKL(\pi_{\vn}\,\|\,\widehat\pi_{\vn})$ admits an exact decomposition into marginal, projection and reduced-density terms and bounds the posterior KL divergence averaged over the true predictive distribution. We derive a priori bounds for these terms, together with estimates or bounds from residual samples. Weighted nearest-neighbour entropy estimates are used for the marginal and reduced-density terms, with asymptotic standard errors in exact rank coordinates. For the projection error, the score is fitted on one subsample, and its mean squared magnitude outside the retained subspace is estimated on the other. Corrections for score and rank-transform errors then give a theoretical upper bound with probability at least $1-\delta$. Tests on three known noise laws in dimension $d=20$ use the estimated rank transform. In $11$ of $12$ combinations of example and sample size, the sum of the two estimated terms and the trace bound exceeds the reference divergence after averaging over the estimated models. In these tests, score-estimation error is the largest contribution to looseness. The overestimate from the Gaussian log-Sobolev inequality alone is at most a factor of $1.98$.
\end{abstract}



% =====================================================================
\section{Introduction}
\label{sec:intro}

Approximating the observation-noise density changes the posterior in a Bayesian inference problem. The noise-density KL divergence bounds the resulting posterior KL divergence, averaged over the true predictive distribution. Part~I \cite{CAS_partI} estimates a non-Gaussian noise density $\pi_{\vn}$ from residual samples, and tests the estimate against known noise laws. Such a comparison is unavailable in applications. We therefore consider the error
\begin{equation}
  \DKL(\pi_{\vn}\,\|\,\widehat\pi_{\vn}),
  \label{eq:target-error}
\end{equation}
of the \textsc{Cas} estimate $\widehat\pi_{\vn}$, estimating its individual contributions from the residual samples.

The decomposition in \Cref{prop:exact-decomp} assigns one error term to each estimation step of \textsc{Cas}:
\eqref{eq:target-error} equals $\mathcal E_{\rm marg} + \mathcal E_{\rm proj} + \mathcal E_{\rm dens}$. The marginal term $\mathcal E_{\rm marg}$ consists of $d$ one-dimensional divergences and two corrections; $\mathcal E_{\rm dens}$ is a divergence between $r$-dimensional laws. For the projection error $\mathcal E_{\rm proj}$, the rank-$r$ reduction admits a bound in terms of the copula-score variance outside the retained subspace. This is the Gaussian-reference bound of certified dimension reduction \cite{ZahmCuiLawSpantiniMarzouk2022}.

A direct entropy estimate of \eqref{eq:target-error} would involve the differential entropy of $\pi_{\vn}$ in dimension $d$. For densities supported on $[0,1]^d$ in Lipschitz balls of smoothness $s=2$, the minimax root-mean-square error has order $(N\log N)^{-2/(d+2)}+N^{-1/2}$ \cite[Thm.~1]{HanJiaoWeissmanWu2020}. The decomposition reduces the required entropy calculations to dimensions $1$ and $r$. No entropy estimate is needed for the bound on $\mathcal E_{\rm proj}$.

The a priori bounds in \Cref{sec:apriori} depend on the rank $r$, the multi-index sets and the sample size $N$. Normalizing the fitted joint density introduces a correction in the marginal bound. The bound minus this correction tends to zero as $N\to\infty$, under the stated assumptions on the marginal estimates and rank-transform error. The bounds on $\mathcal E_{\rm proj}$ and $\mathcal E_{\rm dens}$ keep approximation, regularization and constraint terms that do not vanish as $N\to\infty$. In the examples of \Cref{sec:experiments}, $\mathcal E_{\rm proj}$ is the largest term for $N \ge 12{,}500$, and the Gaussian log-Sobolev bound overestimates it by a factor between $1.64$ and $1.98$. We also apply gradient descent to the tighter bound of \cite{LiCuiLiMarzoukZahm2024sharp}. The subspace moves by $\|\sin\Theta\|_F \le 0.0004$ from the top-$r$ eigenbasis of $\mC$ used by \textsc{Cas}. Most of the remaining looseness of the a priori bound on these examples is due to the Hermite truncation of the score (\Cref{sec:exp-aposteriori}).

For the a posteriori analysis in \Cref{sec:aposteriori}, we estimate $\mathcal E_{\rm marg}$ and $\mathcal E_{\rm dens}$ by weighted nearest-neighbour entropy estimators. Under the stated assumptions these attain the parametric rate in exact rank-Gaussianized coordinates. To bound $\mathcal E_{\rm proj}$, we fit the score and rank transform on half of the sample. The other half estimates the mean squared magnitude of the fitted score outside the retained subspace. Under the stated assumptions, a theoretical upper bound for this quantity holds with probability at least $1-\delta$. It becomes a bound on $\mathcal E_{\rm proj}$ after adding corrections for errors in the score estimate and rank transform. The increment between two nested multi-index sets estimates the Hermite truncation error under a saturation assumption. It does not estimate the sampling or regularization errors.

\Cref{sec:estimator} defines the estimator and decomposes its error. The three numerical examples are reported in \Cref{sec:experiments}.

% =====================================================================
\section{The estimator and its exact error decomposition}
\label{sec:estimator}

The two stages of \textsc{Cas} estimate the density $\pi_{\vn}$ of the noise $\vn\in\RR^d$ from $N$ residual samples (Part~I \cite{CAS_partI}, Algorithm~3.1). For marginal densities $\pi_i$ with CDFs $F_i$, the componentwise map $\scZ(\vx)_i := \Phi^{-1}(F_i(x_i))$ produces standard Gaussian marginals \textup{(2.1)}, with $\Phi$ the standard normal CDF. The transformed noise $\vZ := \scZ(\vn)$ has law $\pi_{\vZ}$. Replacing $F_i$ by the rank-based $\widehat F_i$ of Part~I defines $\widehat\scZ$. The marginal densities are estimated separately by Gaussian-kernel estimates $\widehat\pi_i$ \cite{Silverman1986}.

Let $\gauss_k$ denote the standard Gaussian density on $\RR^k$. The copula density is $c^{Z} := \pi_{\vZ}/\gauss_d$. Its score $\scS := \nabla\log c^{Z}$ has second-moment matrix $\mC := \EE_{\pi_{\vZ}}[\scS\scS^\top]$ \textup{(2.2)}, referred to as the score covariance. With eigenvalues ordered as $\lambda_1(\mC)\ge\cdots\ge\lambda_d(\mC)$, a top-$r$ eigenbasis $\mV_r^{\mC}$ spans the copula active subspace.

Stage~1 fits the copula score by Hermite score matching \cite{Hyvarinen2005} in the span of the gradients of multivariate Hermite polynomials $H_\alpha$. The polynomials are orthonormal in $L^2(\gauss_d)$ and indexed by $\Lambda := \Lambda_{K_1,q_1}$. The set $\Lambda_{K,q}$ contains indices with $1\le|\alpha|_1\le K$ and at most $q$ nonzero entries, except for the linear indices $e_i$ \textup{(3.2)}. This fit defines $\widehat\scS := \nabla\sum_{\alpha\in\Lambda}\widehat\theta^d_\alpha H_\alpha$ \textup{(3.3)}, with coefficient vector $\whvtheta := (\widehat\theta^d_\alpha)_{\alpha\in\Lambda}$. From the empirical second-moment matrix $\hmC$ of $\widehat\scS$ \textup{(3.5)}, we choose a top-$r$ eigenbasis $\hmV_r^{\mC}$. The superscript distinguishes this basis from a principal-component basis of $\vZ$.

A basis $\mV_r\in\St(r,\RR^d)$ and its orthonormal completion $\mV_\perp$ separate the retained coordinates $\vU := \mV_r^\top\vZ$ from the discarded coordinates $\vW := \mV_\perp^\top\vZ$. The conditional law of $\vW$ given $\vU$ is denoted by $\pi_{\vW\mid\vU}$. Define the reduced factor $c_r^{U} := \pi_{\vU}/\gauss_r$ and its logarithm $\scL_r := \log c_r^{U}$. Part~I (\PartIref{subsection}{2.1}) uses this factor to define the reduced model $\pi_{\vn}(\vx;\mV_r) := c_r^{U}(\mV_r^\top\scZ(\vx))\prod_i\pi_i(x_i)$ \textup{(2.4)}. For the fitted estimator, we use $\mV_r = \hmV_r^{\mC}$ in these definitions. 

Stage~2 uses the multi-index set $\Lambda_r := \Lambda_{K_2,q_2}$ in $r$ variables to fit $\scL_r(\vu;\hvbeta) := \sum_{\alpha\in\Lambda_r}\widehat\theta^r_\alpha H_\alpha(\vu)$. The coefficient vector $\hvbeta := (\widehat\theta^r_\alpha)_{\alpha\in\Lambda_r}$ solves the regularized and constrained score-matching problem \textup{(3.9)}. The Gram matrix $\widehat\mA_r$, right-hand side $\widehat\mb_r$ and centering matrix $\widehat\mB$ have population counterparts $\mA^\star$, $\mb^\star$ and $\mB^\star$, respectively, in exact coordinates. Dividing the exponential by $\widehat Z_r := \EE_{\gauss_r}[e^{\scL_r(\cdot;\hvbeta)}]$ produces the reduced factor $\widehat c_r^{U} := e^{\scL_r(\cdot;\hvbeta)}/\widehat Z_r$ and the normalized density $\widehat\pi_{\vU} := \widehat c_r^{U}\gauss_r$. Composing this factor with the fitted rank transform and multiplying by the marginal estimates requires another normalization, by $\widehat Z_{\mathrm{comp}}$. Thus the joint estimate \textup{(3.1)} takes the form
\[
\begin{aligned}
  \widehat\pi_{\vn}(\vx) &:= \widehat Z_{\mathrm{comp}}^{-1}\,\widehat c_r^{U}(\hmV_r^{\mC\top}\widehat\scZ(\vx))\prod_i\widehat\pi_i(x_i),\\
  \widehat Z_{\mathrm{comp}} &:= \int_{\RR^d}\widehat c_r^{U}(\hmV_r^{\mC\top}\widehat\scZ(\vx))\prod_i\widehat\pi_i(x_i)\,d\vx,
\end{aligned}
\]
with copula factor $\widehat c^{Z}(\vz) := \widehat c_r^{U}(\hmV_r^{\mC\top}\vz)$. The normalizing constant $\widehat Z_{\mathrm{comp}}$ equals $1$ when each $\widehat F_i$ is the distribution function of $\widehat\pi_i$. In \textsc{Cas}, $\widehat\scZ$ uses the empirical distribution functions and the $\widehat\pi_i$ are kernel density estimates, so this condition fails. In the notation of Part~I, $\widehat Z_{\mathrm{comp}}=\widehat Z_{\vn}/\widehat Z_r$, since $\widehat c_r^{U}$ already includes the factor $\widehat Z_r^{-1}$. The normalizing constants are evaluated by quasi-Monte Carlo (\PartIref{appendix}{G}). Hats denote estimates and unhatted symbols their exact counterparts. The symbols $\Gr(r,\RR^d)$ and $\St(r,\RR^d)$ denote the Grassmann and Stiefel manifolds \cite{AbsilMahonySepulchre2008}, and $\Theta(\mV,\mV')$ is the vector of principal angles between the spans of two bases \cite{BjorckGolub1973}.

\paragraph{Marginal assumption} Throughout we assume \textup{(M)}: $\EE_{\pi_{\vn}}[|\log(\pi_i/\widehat\pi_i)|] < \infty$ for every $i$, and $\EE[\sum_i\DKL(\pi_i\,\|\,\widehat\pi_i)] = O(\sqrt d\,(\log N)^{K_2/2}N^{-1/2})$. The expectation is taken over repeated estimation samples. The assumed rate requires the marginal density error to decay at least as fast as the rank-transform contribution under \textup{(T)} (\Cref{sec:apriori-marg}). We do not derive this rate for the Gaussian-kernel estimates.

\begin{proposition}[Decomposition of the KL divergence by estimation stage]
\label{prop:exact-decomp}
Let $\widehat\pi_{\vn}$ be the estimator above, with $\vU := \hmV_r^{\mC\top}\vZ$. Under Assumption~2.1 of Part~I and the integrability condition of \textup{(M)}, every term below is finite and, for each realization of the estimation sample,
\begin{equation}
\begin{aligned}
  \DKL(\pi_{\vn}\,\|\,\widehat\pi_{\vn})
  \;=\;
  &\underbrace{\textstyle\sum_{i=1}^d \DKL(\pi_i\,\|\,\widehat\pi_i) \;+\; R_N \;+\; \log\widehat Z_{\mathrm{comp}}}_{\textstyle \mathcal E_{\rm marg}\ \text{(the marginals)}}
  \;+\;
  \underbrace{\DKL\bigl(\pi_{\vn}\,\|\,\pi_{\vn}(\,\cdot\,;\hmV_r^{\mC})\bigr)}_{\textstyle \mathcal E_{\rm proj}\ \text{(the subspace)}}
  \\[4pt]
  &\;+\;
  \underbrace{\DKL\bigl(\pi_{\vU}\,\|\,\widehat\pi_{\vU}\bigr)}_{\textstyle \mathcal E_{\rm dens}\ \text{(the reduced density)}},
\end{aligned}
  \label{eq:kl-three-way}
\end{equation}
where $R_N := \EE_{\pi_{\vn}}[\log\widehat c^{Z}(\scZ)] - \EE_{\pi_{\vn}}[\log\widehat c^{Z}(\widehat\scZ)]$ is the difference between the expected log copula factors evaluated at the exact and estimated rank transforms. The normalizing constant $\widehat Z_{\mathrm{comp}}$ is defined above, and $\mathcal E_{\rm proj} = \EE_{\pi_{\vU}}[\DKL(\pi_{\vW\mid\vU}\,\|\,\gauss_{d-r})]$.
\end{proposition}

\Cref{app:thm-trunc} proves this decomposition at the estimated basis $\hmV_r^{\mC}$, without the eigenbasis of $\mC$. Neither $R_N$ nor $\log\widehat Z_{\mathrm{comp}}$ has a prescribed sign, so $\mathcal E_{\rm marg}$ can be negative.

% =====================================================================
\pagebreak[4]
\section{A priori error bounds}
\label{sec:apriori}

\subsection{The marginals}
\label{sec:apriori-marg}

For a fixed dictionary of degree $K_2$ we assume the moment bounds \textup{(T)},
\begin{equation}
\EE\|\hvbeta\|^2=O(1),\qquad
\EE\,\EE_{\pi_{\vn}}\!\left[(1+\|\scZ\|+\|\widehat\scZ\|)^{2(K_2-1)}
\|\widehat\scZ-\scZ\|^2\right]
=O\!\left(\frac{d(\log N)^{K_2}}{N}\right).
\label{eq:transform-moment}
\end{equation}
where the outer expectation is over the estimation sample and the inner one over a fresh residual. A mean-value bound for the polynomial $\widehat\scL_r$ and the Cauchy--Schwarz inequality give $\EE|R_N|=O(\sqrt d\,(\log N)^{K_2/2}N^{-1/2})$ (\Cref{app:thm-trunc}). Combining this with \textup{(M)} gives
\begin{equation}
  \EE[\mathcal E_{\rm marg} - \log\widehat Z_{\mathrm{comp}}] \;=\; O\!\bigl(\sqrt d\,(\log N)^{K_2/2}N^{-1/2}\bigr),
  \label{eq:Delta-marg}
\end{equation}
and $\mathcal E_{\rm marg} - \log\widehat Z_{\mathrm{comp}} = O_p(\sqrt d\,(\log N)^{K_2/2}N^{-1/2})$. The constants depend on $d$ and the dictionary. We compute the remaining correction $\log\widehat Z_{\mathrm{comp}}$ for each estimate and do not bound it a priori.

\subsection{The subspace}
\label{sec:stage1-error}

The Stage-1 bounds follow from the next inequality, in which $\scS'$ is an arbitrary score field.

\begin{proposition}[Stage-1 bounds]
\label{prop:stage1-cert}
Let $\hmV_r^{\mC}$ be the estimated Stage-1 basis of \textup{(3.5)} and $\hmP_r := \hmV_r^{\mC}\hmV_r^{\mC\top}$. Under Assumption~2.1 of Part~I, for every score field $\scS' \in L^2(\pi_{\vZ};\RR^d)$ with second-moment matrix $\mC' := \EE_{\pi_{\vZ}}[\scS'\scS'^{\top}]$,
\begin{equation}
  \sqrt{2\,\mathcal E_{\rm proj}}
  \;\le\;
  \|\scS - \scS'\|_{L^2(\pi_{\vZ})} \;+\; \sqrt{\tr\!\bigl((\mI-\hmP_r)\mC'\bigr)}.
  \label{eq:finite-sample-kl}
\end{equation}
For $\scS' = \scS$, the score-error term in \eqref{eq:finite-sample-kl} vanishes. Squaring the remaining inequality and dividing by two yields the oracle bound
\begin{equation}
  \mathcal E_{\rm proj} \;=\; \DKL\!\bigl(\pi_{\vn}\,\|\,\pi_{\vn}(\,\cdot\,;\hmV_r^{\mC})\bigr)
  \;\le\;
  C_1 \;:=\; \tfrac12\tr\!\bigl((\mI-\hmP_r)\mC\bigr),
  \label{eq:stage1-cert}
\end{equation}
The oracle bound uses half the trace of $\mC$ over the discarded subspace. For comparison, let $\mP_r^{\mC} := \mV_r^{\mC}\mV_r^{\mC\top}$ project onto the top-$r$ eigenspace of $\mC$. Subtracting and adding the trace at this projector separates the rank and suboptimality contributions:
\begin{equation}
  C_1 \;=\; \underbrace{\tfrac12\textstyle\sum_{i>r}\lambda_i(\mC)}_{\text{rank}} \;+\; \underbrace{\tfrac12\bigl[\tr(\mP_r^{\mC}\mC) - \tr(\hmP_r\mC)\bigr]}_{\Delta_{\rm subopt}\,\ge\,0},
  \label{eq:stage1-split}
\end{equation}
where the rank term is fixed by $\mC$ and $r$. The remaining term $\Delta_{\rm subopt}$ is of order $\|\sin\Theta(\hmV_r^{\mC},\mV_r^{\mC})\|_F^2$. It is zero whenever $\mathrm{span}(\hmV_r^{\mC}) = \mathrm{span}(\mV_r^{\mC})$, and only in this case if $\lambda_r(\mC) > \lambda_{r+1}(\mC)$.
\end{proposition}

To prove \Cref{prop:stage1-cert}, apply the triangle inequality in $L^2(\pi_{\vZ};\RR^d)$ to the Gaussian-reference bound \textup{(2.5)} of Part~I (\Cref{app:thm-trunc}). This bound is \cite[Cor.~2.10]{ZahmCuiLawSpantiniMarzouk2022} specialized to a standard Gaussian reference. In the score-ratio-matching stability result of \cite[Thm.~3.3]{BaptistaBrennanMarzouk2025}, the divergence at the top-$r$ eigenspace of $\mC'$ is bounded by $\|\scS - \scS'\|_{L^2(\pi_{\vZ})}^2 + \sum_{i>r}\lambda_i(\mC')$. This follows from \eqref{eq:finite-sample-kl} at that eigenprojector after applying $(a+b)^2 \le 2a^2 + 2b^2$. Our bound retains the square root and applies to the estimated projector $\hmP_r$.

Projecting $\scS$ onto the gradient space $\scF_\Lambda$ in $L^2(\pi_{\vZ})$ defines $\scS_\Lambda$, its second-moment matrix $\mC_\Lambda$, and the top-$r$ projector $\mP_r^{\Lambda\star}$. Because $\hmC$ is computed from the penalized estimate of Part~I, its limiting projector depends on the penalty. If the penalty strength tends to zero, the projector of $\hmC$ converges to $\mP_r^{\Lambda\star}$ under the consistency and eigengap conditions of \Cref{lem:score-cov-rate}. The experiments keep the penalty fixed. In that case the limit is the top-$r$ projector of $\mC_{\Lambda,\tau}$, the second-moment matrix of the penalized population estimate, provided it has an eigengap. The choice $\scS'=\scS_\Lambda$ in \eqref{eq:finite-sample-kl}, followed by the trace split, gives the next bound.

\begin{theorem}[Stage-1 KL divergence contribution]
\label{thm:stage1-rate}
Let $\hmV_r^{\mC}$ be the estimated Stage-1 basis of \textup{(3.5)}, let $\hmP_r := \hmV_r^{\mC}\hmV_r^{\mC\top}$, and let $\lambda_1(\mC_\Lambda) \ge \lambda_2(\mC_\Lambda) \ge \cdots$ be the eigenvalues of $\mC_\Lambda$. Under Assumption~2.1 of Part~I,
\begin{equation}
\begin{aligned}
  \sqrt{2\,\DKL\!\bigl(\pi_{\vn}\,\|\,\pi_{\vn}(\,\cdot\,;\hmV_r^{\mC})\bigr)}
  \;\le\;
  &\underbrace{\vphantom{\sqrt{\textstyle\sum_{i>r}\lambda_i(\mC_\Lambda)}}\|\scS - \scS_\Lambda\|_{L^2(\pi_{\vZ})}}_{\text{Hermite multi-index truncation}}
  \;+\;
  \underbrace{\sqrt{\textstyle\sum_{i>r}\lambda_i(\mC_\Lambda)}}_{\text{subspace rank truncation}}
  \\[2pt]
  &\;+\;
  \underbrace{\sqrt{\tr\!\bigl((\mP_r^{\Lambda\star}-\hmP_r)\mC_\Lambda\bigr)}}_{\text{subspace statistical estimation}}.
\end{aligned}
  \label{eq:stage1-threeterm}
\end{equation}
If the moment and regularity conditions \textup{(A1)}--\textup{(A3)} of \Cref{app:sample-proof} hold as well, then the third term satisfies
\begin{equation}
  \tr\!\bigl((\mP_r^{\Lambda\star}-\hmP_r)\mC_\Lambda\bigr)
  \;\le\;
  2r\,\|\hmC - \mC_\Lambda\|_{\rm op}
  \;=\;
  O_p(N^{-1/2}) \;+\; O(\tau_N)
  \label{eq:stage1-estrate}
\end{equation}
at the Tikhonov regularization strength $\tau_N$ of \textup{(A3)} \cite{EnglHankeNeubauer1996}.
\end{theorem}

See \Cref{app:thm-trunc} for the proof. In \eqref{eq:stage1-threeterm}, the first term is the score approximation error in $\Lambda$, the second is the rank-$r$ truncation error of $\mC_\Lambda$, and the third is the subspace error due to sampling and regularization. Equation~\eqref{eq:stage1-estrate} bounds the third term by $O_p(N^{-1/2})+O(\tau_N)$, and the $O(\tau_N)$ part does not vanish at a fixed ridge strength.

A nonzero Hermite truncation error does not by itself prevent recovery of the subspace. If $\mC$ has rank $r$, the interaction order is unrestricted ($q_1\ge\min(K_1,d)$) and $\mC_\Lambda$ also has rank $r$, then the top-$r$ eigenspace of $\mC_\Lambda$ is the active subspace, regardless of $\|\scS-\scS_\Lambda\|_{L^2(\pi_{\vZ})}$. Part~I \cite{CAS_partI} measures the departure from this case for the multi-index sets used in practice.

\subsection{The reduced density}
\label{sec:stage2-error}

For Stage~2, the KL term in \eqref{eq:kl-three-way} is bounded by applying a log-Sobolev inequality to the Fisher-divergence bound in \Cref{thm:ccsm-rate}.

\begin{theorem}[Stage-2 KL divergence contribution]
\label{thm:stage2-kl}
Let $\widehat\pi_{\vU}$ be an exactly normalized Stage-2 density satisfying a logarithmic Sobolev inequality with $C_{\rm LSI}(\hvbeta)>0$, as ensured by the global constraints in \Cref{prop:lsi-holds}. Let $\scL_r^\star$ be the Fisher projection over $\Lambda_r$, and $\mathcal J_{\rm est}(\hvbeta):=\|\nabla(\scL_r^\star-\scL_r(\cdot;\hvbeta))\|_{L^2(\pi_{\vU})}^2$. Under \textup{(H1)--(H2)} and the population score-matching regularity of \Cref{app:bisc-ridge}, for every realization,
\begin{equation}
\DKL(\pi_{\vU}\,\|\,\widehat\pi_{\vU})
\le\frac{\mathcal J_{\rm trunc}+\mathcal J_{\rm est}(\hvbeta)}{2C_{\rm LSI}(\hvbeta)}.
\label{eq:stage2-kl}
\end{equation}
Here $\mathcal J_{\rm trunc}\le C_\star K^{-(2s-1)}\|\scL_r\|_{H^s_{\gauss}}^2$, with
$K:=\min\{|\alpha|_1\ge1:\alpha\notin\Lambda_r,\ \vartheta_\alpha\ne0\}$, and $\mathcal J_{\rm trunc}=0$ when this set is empty. Under the hypotheses of \Cref{thm:ccsm-rate} and \textup{(H3)}: $C_{\rm LSI}(\hvbeta)\ge c_{\rm LSI}>0$ almost surely,
\begin{equation}
\EE[\DKL(\pi_{\vU}\,\|\,\widehat\pi_{\vU})]
\le\mathcal B_{\rm dens}+O(N^{-1}),\qquad
\mathcal B_{\rm dens}:=
\frac{\mathcal J_{\rm trunc}+\mathcal E_{\rm est}+16\EE[d_N^2\mathbf1_{\mathcal G}]}{2c_{\rm LSI}}.
\label{eq:stage2-kl-mean}
\end{equation}
where $\mathcal E_{\rm est}$, the event $\mathcal G$ and the constraint distance $d_N$ are those of \Cref{thm:ccsm-rate}. For constants chosen from a finite grid, assume that every pair satisfies the hypotheses with a common $c_{\rm LSI}$. Summing $\mathcal E_{\rm est}$ and the constraint term over the grid makes \eqref{eq:stage2-kl-mean} valid for any selection rule.
\end{theorem}

\Cref{app:excess-fisher} proves the first bound using the log-Sobolev inequality for $\widehat\pi_{\vU}$, and the expectation bound then follows from \Cref{thm:ccsm-rate}.

\begin{proposition}[The constrained estimate satisfies a log-Sobolev inequality]
\label{prop:lsi-holds}
Let $\widehat\pi_{\vU} = \widehat Z_r^{-1}\exp\bigl(\textstyle\sum_{\alpha\in\Lambda_r}\widehat\theta^r_\alpha H_\alpha\bigr)\,\gauss_r$, and let $K_{\rm eff} \ge 4$ be the largest even integer at most $K_2$. Assume the coefficients satisfy the constraints of the Stage-2 solve (\PartIref{appendix}{E} of Part~I). In particular, coefficients of degree $K_2$ vanish when $K_2$ is odd, and the leading form $F[\hvbeta](\vv) := \sum_{|\alpha|_1 = K_{\rm eff}}\widehat\theta^r_\alpha\,\vv^\alpha/\sqrt{\alpha!}$ satisfies $F[\hvbeta](\vv) \le -\varepsilon$ for every $\|\vv\| = 1$. Then $\widehat\pi_{\vU}$ is a probability measure and satisfies a logarithmic Sobolev inequality with a constant $C_{\rm LSI} > 0$.
\end{proposition}

\Cref{app:lsi-proof} establishes a positive log-Sobolev constant for each fitted density. A common lower bound over all admissible coefficients, an unbounded set, must still be assumed in \textup{(H3)}. The constraint distance $d_N$ is positive when $\bm\theta^{r,\star}$ violates the empirical centering constraint or the Stage-2 level bound. We bound neither $d_N$ nor $c_{\rm LSI}$, and \Cref{sec:aposteriori} estimates $\mathcal E_{\rm dens}$ directly.

The best reduced-score approximation in $\Lambda_r$ has Fisher error $\mathcal J_{\rm trunc}$. This error decreases as $\Lambda_r$ is enlarged and tends to zero if $\nabla\scL_r$ is in the $L^2(\pi_{\vU})$-closure of the union of the gradient spaces. For the multi-index sets of Part~I, the linear modes $e_j$ are excluded. The bound in \textup{(H2)} therefore tends to zero only if the linear Hermite coefficients $\vartheta_{e_j}$ of $\scL_r$ vanish. We select the Stage-2 degree by cross-validation \cite[\S 7.10]{HastieTibshiraniFriedman2009}.

\subsection{The overall a priori bound}
\label{sec:composite-kl}

\begin{proposition}[A priori bound]
\label{prop:apriori-rates}
Compute the Stage-2 coefficients from exact-transform observations independent of the sample used to compute $\hmV_r^{\mC}$. Assume Assumption~2.1 of Part~I, \textup{(M)}, \textup{(T)}, and the hypotheses of \Cref{thm:stage2-kl,thm:ccsm-rate} for every basis $\mV\in\St(r,\RR^d)$. All constants must be independent of $\mV$, including the lower bounds on $\lambda_{\min}(\mA^\star)$ and on the log-Sobolev constant. Then
\begin{equation}
\begin{aligned}
  \EE\bigl[\DKL(\pi_{\vn}\,\|\,\widehat\pi_{\vn})\bigr]
  \;\le\;
  &\underbrace{\EE[\log\widehat Z_{\mathrm{comp}}] + O\!\bigl(\sqrt d\,(\log N)^{K_2/2}N^{-1/2}\bigr)}_{\text{marginals}}
  \\[2pt]
  &\;+\;
  \underbrace{\tfrac12\textstyle\sum_{i>r}\lambda_i(\mC) + \EE[\Delta_{\rm subopt}]}_{\text{Stage 1}}
  \;+\;
  \underbrace{\EE\bigl[\mathcal B_{\rm dens}(\hmV_r^{\mC})\bigr] + O(N^{-1})}_{\text{Stage 2}},
\end{aligned}
  \label{eq:composite-kl}
\end{equation}
where $\mathcal B_{\rm dens}(\mV)$ is the bound $\mathcal B_{\rm dens}$ of \eqref{eq:stage2-kl-mean} for the basis $\mV$.
\end{proposition}

The proof in \Cref{app:thm-trunc} relies on the independent sample and exact rank transform used in Stage~2 of this idealized estimator. Conditionally on the Stage-1 sample, the subspace is fixed and the Stage-2 observations are independent draws from the true projected distribution.

In the Part~I implementation, both stages use the same sample and an estimated rank transform. For each fitted density, \eqref{eq:kl-three-way} and \eqref{eq:stage2-kl} remain valid under their stated assumptions. The expected-error bound \eqref{eq:composite-kl} has not been proved for this implementation.

The marginal bound, apart from $\EE[\log\widehat Z_{\mathrm{comp}}]$, tends to zero as $N\to\infty$. In Stage~1, the trailing eigenvalue sum of $\mC$ decreases with $r$, whereas the estimated subspace's suboptimality depends on $\Lambda$ and the Stage-1 ridge. The Stage-2 bound keeps the truncation term, the curvature part of the regularization bias and the constraint term as $N$ grows.

% =====================================================================
\section{A posteriori error estimation}
\label{sec:aposteriori}

Some quantities in \Cref{sec:apriori} cannot be computed from residual samples alone. Direct sample estimates are available for $\mathcal E_{\rm marg}$ and $\mathcal E_{\rm dens}$. For $\mathcal E_{\rm proj}$, the sample first gives a bound on the mean squared magnitude of the fitted score outside the retained subspace. A bound on the projection error also requires corrections for score estimation and the rank transform.

\subsection{The marginals and the reduced density}
\label{sec:apost-density}

The identity $\mathcal E_{\rm dens} = -H(\pi_{\vU}) - \EE_{\pi_{\vU}}[\log\widehat\pi_{\vU}(\vU)]$ separates entropy from an expectation of the fitted log-density. The marginal divergence $\DKL(\pi_i\,\|\,\widehat\pi_i)$ likewise uses $H(\pi_i)$ and $\EE_{\pi_i}[\log\widehat\pi_i]$. The fitted log-densities have closed forms, so their contributions are estimated by negative sample means. Estimating the entropies is a nonparametric problem in dimension $r$ or $1$.

Let $\vx_1,\dots,\vx_M$ be residuals independent of the estimation sample, with projected rank-Gaussianized coordinates $\vu_m:=\hmV_r^{\mC\top}\scZ(\vx_m)$. Apply the weighted nearest-neighbour entropy estimator \cite[\S2]{BerrettSamworthYuan2019} to $\vu_1,\dots,\vu_M$, and denote the result by $\widehat H^w_M$. The choices of the number of neighbours $k$ and sample size $M$ are given in \Cref{sec:exp-endtoend}. The two error estimates are
\begin{equation}
  \begin{aligned}
  \widehat{\mathcal E}_{\rm dens} &:= -\,\widehat H^{w}_M - \frac1M\sum_{m=1}^M \log\widehat\pi_{\vU}(\vu_m),\\
  \widehat{\mathcal E}_{\rm marg} &:= \sum_{i=1}^d\Bigl(-\,\widehat H_{M,i} - \frac1M\sum_{m=1}^M\log\widehat\pi_i(x_{m,i})\Bigr) + \log\widehat Z_{\mathrm{comp}},
  \end{aligned}
  \label{eq:apost-est}
\end{equation}
where $\widehat H_{M,i}$ is computed with the same estimator at $r=1$, using coordinate $i$ of $\vx_1,\dots,\vx_M$. The target of the marginal estimate is $\mathcal E_{\rm marg}-R_N$: the term $R_N$ is omitted. Evaluating $\log\widehat\pi_{\vU}$ entails computing $\widehat Z_r$, and $\widehat{\mathcal E}_{\rm marg}$ additionally uses $\widehat Z_{\mathrm{comp}}$. We compute $\widehat Z_r$ with $M_{\rm norm}=8192$ scrambled Sobol' points (Supplement~\ref{sm:experimental-protocol}).
For the convergence rate below, the highest-degree part of the fitted polynomial must be strictly negative in every direction. This is the global Stage-2 constraint.
For each fixed fit, Sobol' quadrature of the exact integrand then gives an error of $O(M_{\rm norm}^{-1+\epsilon})$ in $\log\widehat Z_r$, for every $\epsilon>0$, as the number of points increases through powers of two.
\Cref{app:qmc-normalizer} gives the proof, which uses the Koksma--Hlawka inequality \cite[Thm.~2.11]{Niederreiter1992}.
The implicit constant depends on the fitted polynomial, so the rate alone gives no numerical error bound for $8192$ points. The central limit theorem of \Cref{cor:apost-den} uses the exact normalizer, and it holds with a computed normalizer when the error in $\log\widehat Z_r$ is $o_p(M^{-1/2})$.

\begin{corollary}[A posteriori estimate of the reduced-density and marginal divergences]
\label{cor:apost-den}
Assume, conditionally on the estimation sample: \textup{(D1)} $\pi_{\vU}$ belongs to the class $\mathcal F_{r,\theta}$ of \cite[\S2]{BerrettSamworthYuan2019}, with $\alpha>r$ moments and smoothness $\beta>r/2$; \textup{(D2)} $\Var_{\pi_{\vU}}\log\pi_{\vU}<\infty$ and $\Var_{\pi_{\vU}}\log\widehat\pi_{\vU}<\infty$; \textup{(D3)} $\vx_1,\dots,\vx_M$ are i.i.d.\ $\pi_{\vn}$ and independent of the estimation sample. Take $w^{(k)}\in\mathcal W^{(k)}$ with $\sup_k\|w^{(k)}\|<\infty$ and $k=k_M$ as in \cite[Thm.~1]{BerrettSamworthYuan2019}. Then $\widehat{\mathcal E}_{\rm dens}\to\mathcal E_{\rm dens}$ in probability, and
\begin{equation}
  \sqrt M\,\bigl(\widehat{\mathcal E}_{\rm dens} - \mathcal E_{\rm dens}\bigr)
  \;\xrightarrow{\;d\;}\;
  \mathcal N\bigl(0,\ \sigma_{\rm den}^2\bigr),
  \qquad
  \sigma_{\rm den}^2 \;=\; \Var_{\pi_{\vU}}\!\Bigl[\log\frac{\pi_{\vU}(\vU)}{\widehat\pi_{\vU}(\vU)}\Bigr],
  \label{eq:apost-clt}
\end{equation}
the variance of the oracle log-ratio estimator. For each marginal divergence, the corresponding one-dimensional hypotheses give the same conclusion. Minkowski's inequality also gives
\[
\sigma_{\rm den}\le\sqrt{\Var_{\pi_{\vU}}\log\pi_{\vU}}
+\sqrt{\Var_{\pi_{\vU}}\log\widehat\pi_{\vU}}.
\]
The variance estimator of \cite[\S2]{BerrettSamworthYuan2019} and the sample variance of $\log\widehat\pi_{\vU}(\vu_m)$ estimate the first and the second variance consistently. We report the sum of their square roots divided by $\sqrt M$ as an estimate of an upper bound on the asymptotic standard error.
\end{corollary}

The proof is in \Cref{app:apost-den}.

Condition \textup{(D1)} restricts the smoothness and moments of the unknown reduced density, but permits unbounded support. The coordinates of $\vU$ are linear combinations of standard Gaussian variables, without any requirement of independence. The polynomial $\log\widehat\pi_{\vU}$ therefore has all moments under $\pi_{\vU}$, which implies the second part of \textup{(D2)}.

The corollary concerns exact rank-Gaussianized coordinates. In applications and in \Cref{sec:exp-endtoend}, held-out residuals instead pass through $\widehat\scZ$, which is constant beyond the estimation-sample range. Thus the law of $\hmV_r^{\mC\top}\widehat\scZ(\vx)$, $\vx\sim\pi_{\vn}$, has atoms at the images of the outer orthants. This law has no density and does not satisfy the corollary's hypotheses.

Suppose a different transform gives an image law $Q_N$ with a density, and that the following entropies and expectations are finite. An entropy estimate consistent for $Q_N$ then yields a reduced-density error estimate for $\DKL(Q_N\,\|\,\widehat\pi_{\vU})$. Denote by $P$ the law of $\hmV_r^{\mC\top}\scZ(\vx)$, $\vx\sim\pi_{\vn}$, and define $R_N$ for this transform as in \Cref{prop:exact-decomp}. The discrepancy between the two targets is
\begin{equation}
  \DKL(Q_N\,\|\,\widehat\pi_{\vU}) - \mathcal E_{\rm dens}
  \;=\; H(P) - H(Q_N) \;+\; R_N \;+\; \tfrac12\bigl(\EE_{Q_N}\|\vU\|_2^2 - \EE_P\|\vU\|_2^2\bigr),
  \label{eq:fitted-target}
\end{equation}
because $\log\widehat\pi_{\vU} = \log\widehat c_r^{U} + \log\gauss_r$ and $\widehat c^{Z}(\vz) = \widehat c_r^{U}(\hmV_r^{\mC\top}\vz)$. If this transform satisfies \textup{(T)}, the condition controls $R_N$ and the difference of second moments, and the entropy difference $H(P)-H(Q_N)$ needs a separate argument. Equation~\eqref{eq:fitted-target} does not apply to the implemented transform, whose image law has atoms. \Cref{sec:exp-endtoend} compares the implemented estimates with reference values.

\subsection{The subspace}
\label{sec:apost-sub}

By \Cref{prop:exact-decomp}, $\mathcal E_{\rm proj} = \EE_{\pi_{\vU}}[\DKL(\pi_{\vW\mid\vU}\,\|\,\gauss_{d-r})]$. A direct estimate would need an entropy in dimension $d - r$, which is $16$ in our examples. Instead, the Gaussian-reference inequality \textup{(2.5)}, evaluated at the estimated basis, gives $\mathcal E_{\rm proj} \le \tfrac12\tr((\mI-\hmP_r)\mC)$. We fit the score and the rank transform on one half of the sample and use the other half to bound the mean squared magnitude of the fitted score outside the subspace. The bound \eqref{eq:aposteriori-corrected} on $\mathcal E_{\rm proj}$ adds corrections for the score estimation and the rank transform.

\begin{corollary}[Finite-sample bound on the trace of the estimated score]
\label{cor:aposteriori}
Divide the $N$ samples into equal halves. Fit the rank transform and the Stage-1 coefficients $\whvtheta$ on the first half, and write $\vZ^{(1)},\dots,\vZ^{(N_2)}$, $N_2 := N/2$, for the transformed second half. Evaluate $\widehat\scS(\vZ^{(n)})$ to form the empirical second-moment matrix $\hmC$ in \textup{(3.5)}. From this matrix, compute the top-$r$ eigenprojector $\hmP_r$ and trailing eigenvalue sum $\sum_{i>r}\widehat\lambda_i(\hmC)$. Also compute $\widehat m_4 := N_2^{-1}\sum_n\|\widehat\scS(\vZ^{(n)})\|_2^4$.
Conditional on the first half, let $\vZ$ be an independent residual transformed by the fitted rank map. The corresponding population moments are $\widetilde\mC(\whvtheta) := \EE[\widehat\scS(\vZ)\,\widehat\scS(\vZ)^\top]$ and $m_4 := \EE[\|\widehat\scS(\vZ)\|_2^4]$, both conditional on the first half. The range of the fitted transform is bounded, so the fixed polynomial $\widehat\scS$ has a finite fourth moment for every noise law. For any $\delta\in(0,1)$, with probability at least $1-\delta$ over the second half,
\begin{equation}
  \tr\!\bigl((\mI-\hmP_r)\,\widetilde\mC(\whvtheta)\bigr)
  \;\le\;
  \textstyle\sum_{i>r}\widehat\lambda_i(\hmC) \;+\; \sqrt{d-r}\,\sqrt{\frac{m_4}{N_2\,\delta}}.
  \label{eq:aposteriori}
\end{equation}
Replacing the unknown $m_4$ by its conditionally unbiased estimate $\widehat m_4$ changes the probability statement: \eqref{eq:aposteriori} holds with probability at least $1-\delta-o(1)$ over both halves as $N_2\to\infty$, provided $m_8/m_4^2=O_p(1)$ on $m_4>0$, where $m_8:=\EE[\|\widehat\scS(\vZ)\|_2^8]$. On $m_4=0$, the score is almost surely zero, as are both sides of \eqref{eq:aposteriori}.
\end{corollary}

The proof is in \Cref{app:thm-trunc}.

With $\widehat m_4$, the right-hand side of \eqref{eq:aposteriori} can be computed from the sample. Its left-hand side is a trace of the estimated score under the first-half transform. In exact coordinates, the second-moment matrix is $\mC(\whvtheta):=\EE_{\pi_{\vZ}}[\widehat\scS\widehat\scS^\top]$. Applying \eqref{eq:finite-sample-kl} with $\scS'=\widehat\scS$ gives
\begin{equation}
  \sqrt{2\,\mathcal E_{\rm proj}}
  \;\le\;
  \|\scS - \widehat\scS\|_{L^2(\pi_{\vZ})} \;+\; \sqrt{\tr\!\bigl((\mI-\hmP_r)\,\mC(\whvtheta)\bigr)},
  \label{eq:aposteriori-kl}
\end{equation}
whose first term includes the score component orthogonal to $\scF_\Lambda$. This component cannot be determined from a fit in $\scF_\Lambda$ alone. A second fit in a larger space permits estimation of this approximation error, subject to the following assumption.

\begin{proposition}[Hierarchical estimate of the Hermite truncation term]
\label{prop:hier-trunc}
Let $\Lambda \subset \Lambda'$ be Stage-1 multi-index sets, so $\scF_\Lambda \subseteq \scF_{\Lambda'}$. Write $\scS_\Lambda, \scS_{\Lambda'}$ for the $L^2(\pi_{\vZ})$ projections of $\scS$ onto these spaces, and set $\eta := \|\scS_{\Lambda'} - \scS_\Lambda\|_{L^2(\pi_{\vZ})}$. Orthogonality gives $\|\scS - \scS_\Lambda\|^2_{L^2(\pi_{\vZ})} = \|\scS - \scS_{\Lambda'}\|^2_{L^2(\pi_{\vZ})} + \eta^2$. Under the saturation assumption \textup{(SAT)}: $\|\scS - \scS_{\Lambda'}\|_{L^2(\pi_{\vZ})} \le \beta\,\|\scS - \scS_\Lambda\|_{L^2(\pi_{\vZ})}$ for some $\beta \in [0,1)$,
\begin{equation}
  \eta^2 \;\le\; \|\scS - \scS_\Lambda\|^2_{L^2(\pi_{\vZ})} \;\le\; \frac{\eta^2}{1-\beta^2}.
  \label{eq:hier-trunc}
\end{equation}
In exact rank-Gaussianized coordinates, estimate the coefficient differences between the two sets on two disjoint subsamples, obtaining $\Delta\whvtheta_1$ and $\Delta\whvtheta_2$. On a third subsample, compute the Gram matrix $\widehat\mA_{\Lambda'}$ for the gradients $\nabla H_\alpha$, $\alpha\in\Lambda'$. If all three subsample sizes are proportional to $N$ and both sets satisfy \textup{(A1)}--\textup{(A3)}, the cross-fitted estimate \cite{ChernozhukovEtAl2018} $\widehat\eta_\times^2:=(\Delta\whvtheta_1)^\top\widehat\mA_{\Lambda'}\Delta\whvtheta_2$ satisfies $\widehat\eta_\times^2=\eta^2+O_p(N^{-1/2})+O(\tau_N)$.
\end{proposition}

The proof is in \Cref{app:thm-trunc}.

Assumption \textup{(SAT)} requires the larger space to reduce the $L^2(\pi_{\vZ})$ approximation error by a factor no greater than $\beta<1$. This is the saturation condition used in hierarchical a posteriori finite-element estimation \cite{BankWeiser1985,DorflerNochetto2002}. It cannot be verified from the two fitted scores. In our examples, knowledge of the true score permits direct computation of $\beta$ (\Cref{sec:exp-aposteriori}). For the fitted score, the triangle inequality gives $\|\scS-\widehat\scS\|_{L^2(\pi_{\vZ})}\le\eta/\sqrt{1-\beta^2}+\|\scS_\Lambda-\widehat\scS\|_{L^2(\pi_{\vZ})}$, with the last term of order $O_p(N^{-1/2})+O(\tau_N)$ under \textup{(A1)}--\textup{(A3)}.

Set $e_N:=\|\scS-\widehat\scS\|_{L^2(\pi_{\vZ})}$ and $D_N:=\|\mC(\whvtheta)-\widetilde\mC(\whvtheta)\|_F$. These are the score-estimation error and the discrepancy between the fitted score's second moments under the exact and first-half transforms. Use $b_N^{\rm pop}$ for the right-hand side of \eqref{eq:aposteriori}, and $b_N$ for the value obtained by replacing $m_4$ with $\widehat m_4$. By Cauchy--Schwarz in the Frobenius inner product, $|\tr((\mI-\hmP_r)(\mC(\whvtheta)-\widetilde\mC(\whvtheta)))|\le\sqrt{d-r}\,D_N$. Substituting this inequality and \eqref{eq:aposteriori} into \eqref{eq:aposteriori-kl}, conditionally on the first half, gives with probability at least $1-\delta$ over the second half
\begin{equation}
\mathcal E_{\rm proj}\le\tfrac12\left[e_N+
\sqrt{b_N^{\rm pop}+\sqrt{d-r}\,D_N}\right]^2.
\label{eq:aposteriori-corrected}
\end{equation}
The substitution of $b_N$ for $b_N^{\rm pop}$ in \eqref{eq:aposteriori-corrected} gives the asymptotic probability statement of \Cref{cor:aposteriori}, under its moment-ratio condition. It does not give a finite-sample probability of $1-\delta$. The bound on $e_N$ after \Cref{prop:hier-trunc} needs the increment $\eta$ and the saturation constant $\beta$. Its last term is $O_p(N^{-1/2})+O(\tau_N)$, and the $O(\tau_N)$ part does not vanish at a fixed penalty strength. For the transform discrepancy, assume \textup{(T)} with Stage-1 degree $K_1$ and coefficients $\whvtheta$ in place of $K_2$ and $\hvbeta$. The mean-value and Cauchy--Schwarz argument used for $R_N$ then proves $D_N\to0$ in probability. The numerical values of $e_N$ and $D_N$ in \Cref{sec:exp-aposteriori} are computed using the known score.

\subsection{The overall estimate}
\label{sec:apost-assembled}

From $N$ residuals for the estimation and $M$ held-out residuals, the a posteriori estimate of \eqref{eq:target-error} is
\begin{equation}
  \widehat{\mathcal E}_{\rm marg}
  \;+\;
  \tfrac12\Bigl[\textstyle\sum_{i>r}\widehat\lambda_i(\hmC) + \sqrt{d-r}\,\sqrt{\widehat m_4/(N_2\delta)}\Bigr]
  \;+\;
  \widehat{\mathcal E}_{\rm dens}
  \label{eq:apost-assembled}
\end{equation}
Equation~\eqref{eq:apost-assembled} is our a posteriori estimate of the KL divergence \eqref{eq:target-error}. The middle term, $\tfrac12 b_N$, replaces $m_4$ in $\tfrac12 b_N^{\rm pop}$ by the sample estimate $\widehat m_4$. It leaves out the score error $e_N$ and the second-moment discrepancy $D_N$ of \eqref{eq:aposteriori-corrected}. The overall estimate also leaves out the rank-transform correction $R_N$. For the latter, \textup{(T)} gives $\EE|R_N|=O(\sqrt d\,(\log N)^{K_2/2}N^{-1/2})$ (\Cref{sec:apriori-marg}). Because of these omissions, the estimate can fall below the KL divergence. The probability statement for the score trace does not apply to the total. \Cref{sec:exp-endtoend} compares this total with reference values.

% =====================================================================
\section{Sharpness of the bounds}
\label{sec:tightness}

Minimizing \textup{(2.5)} selects the top-$r$ eigenspace of $\mC$. The divergence $\DKL(\pi_{\vn}\,\|\,\pi_{\vn}(\,\cdot\,;\mV_r))$ has no closed form, and its minimizing subspace differs in general from the minimizer of this bound. The experiments in \Cref{sec:experiments} distinguish the looseness of this Gaussian-reference bound from the error in approximating the score by Hermite polynomials. For the former, we compare with the tighter bound on the same divergence in \cite{LiCuiLiMarzoukZahm2024sharp}.

\paragraph{Looseness}
Using the quadrature in Supplement~\ref{sm:disintegration-mc}, we evaluate the divergence at $\mV_r^{\mC}$ in the rank-Gaussianized coordinates used in \textup{(2.5)}. The bound $\tfrac12\sum_{i>r}\lambda_i(\mC)$ exceeds this divergence by factors of $1.91$, $1.64$ and $1.98$ on Examples~1, 2 and~3. Table~\ref{tab:sm-jhi} in the supplement lists both quantities and the Monte Carlo standard error of the divergence.

At the estimated subspace $\hmV_r^{\mC}$ of Example~1, the same quadrature gives the divergence for comparison with $C_1$ in \eqref{eq:stage1-cert}. The ratio of the bound to the divergence is between $1.38$ and $1.45$ at every $N \ge 2{,}500$.

\paragraph{The two subspaces}
Denote the tighter bound by $J_{\mathrm{hi}}$. As a function of the basis, it satisfies $J_{\mathrm{hi}}\le\tfrac12\tr((\mI-\mV_r\mV_r^\top)\mC)$ for every $\mV_r$ (Supplement~\ref{sm:sharper-bound}). At $\mV_r^{\mC}$, replacing \textup{(2.5)} by $J_{\mathrm{hi}}$ reduces the bound by less than $0.5\%$ on each problem. We run Riemannian gradient descent \cite{AbsilMahonySepulchre2008} on $J_{\mathrm{hi}}$, initialized at the top-$r$ eigenbasis of $\mC$. The line search terminates when it finds no decrease in $J_{\mathrm{hi}}$, leaving a basis within $\|\sin\Theta\|_F \le 0.0004$ of the initial basis. Estimates of $\mC$ from $10^6$ and $2\cdot10^6$ samples, by comparison, have top-$r$ eigenbases separated by at least $0.0012$ on every example. The descent therefore changes the \textsc{Cas} subspace less than these two Monte Carlo estimates differ.

\paragraph{Agreement of the two bounds}
The difference between the two functionals is $\tfrac12\sum_j[\lambda_j(F^Z_\perp)-1-\log\lambda_j(F^Z_\perp)]$, where $F^Z_\perp$ is the Fisher information of $\pi_{\vZ}$ in the complement of the retained subspace. The difference is zero when $F^Z_\perp=\mI$. The rotation in our examples leaves twelve coordinates independent and standard Gaussian (\Cref{sec:exp-benchmark}). For a basis supported on the mixed block, these coordinates contribute twelve unit eigenvalues to $F^Z_\perp$. Its estimated eigenvalues lie in $[0.995,1.024]$ (Supplement~\ref{sm:sharper-bound}), which explains the close agreement of the bounds on these examples.

% =====================================================================
\section{Numerical results}
\label{sec:experiments}

The experiments in Part~I \cite{CAS_partI} compare the estimator with Gaussian-copula, product-of-marginals, and PCA-subspace baselines. Here we use the known noise laws to compare the bounds and estimates of \Cref{sec:apriori,sec:aposteriori} with reference values.

\subsection{The three example problems}
\label{sec:exp-benchmark}

Each of the three examples in Part~I \cite{CAS_partI} has dimension $d = 20$, with four coupled latent coordinates. A random rotation acts on these and four independent standard Gaussian coordinates, forming a block of size $b = 8$. The other $12$ coordinates remain independent standard Gaussians. Two coupled coordinates are functions of the other two, with additive Gaussian output noise. The resulting noise density and copula score are smooth. For Examples~1 and~3, these functions are the real and imaginary parts of $(W_0+iW_1)^2$ and $(W_0+iW_1)^3$, respectively, saturated by $\tanh$. Example~2 uses an even function of each input. We take $r = 4$ to match the number of coupled coordinates. Reduction to those four latent coordinates is exact. The rank transform of the mixed coordinates breaks this, and the reference projection error at rank four is nonzero (Supplement~\ref{sm:disintegration-mc}). For all examples, Stage~1 uses $(K_1,q_1) = (4,2)$. The fraction of the copula score represented by this gradient space is reported in \Cref{sec:exp-aposteriori}.

\subsection{Reference values}
\label{sec:exp-reference}

Computing the projection and reduced-density terms in \eqref{eq:kl-three-way} requires $\pi_{\vU}$. The noise laws are known, but the integral over the $d-r$ discarded coordinates has no closed form in these examples. We first tabulate the marginal distribution functions and densities of the mixed coordinates. Gauss--Hermite quadrature then integrates over the complement, and Monte Carlo averaging over the retained coordinates gives the reference divergences and their standard errors (Supplement~\ref{sm:disintegration-mc}).

We compute the reference basis $\mV_r^{\mC}$ for \Cref{sec:tightness} from the top-$r$ eigenspace of $\mC$, estimated using the copula score of Supplement~\ref{sm:disintegration-mc} and $N_{\rm ref} = 10^6$ draws. We also fit Stage~1 to $N_{\rm ref} = 10^6$ residuals at $(K_1,q_1) = (4,2)$ with the regularization of Part~I. This fit gives numerical approximations to the penalized population second-moment matrix $\mC_{\Lambda,\tau}$ and its top-$r$ eigenbasis $\mV_{r,\Lambda,\tau}$. We use this reference basis for the reduced-density terms in \Cref{sec:exp-composite,sec:exp-endtoend}.

\subsection{The three terms of the decomposition}
\label{sec:exp-composite}

For Example~1, \Cref{fig:composite-main} gives reference values of the three terms in \eqref{eq:kl-three-way} for ten estimated models at each sample size. The corresponding results for Examples~2 and~3 are in Supplement~\ref{sm:composite-anatomy}.

\begin{figure*}[t]
\centering
\includegraphics[width=\textwidth]{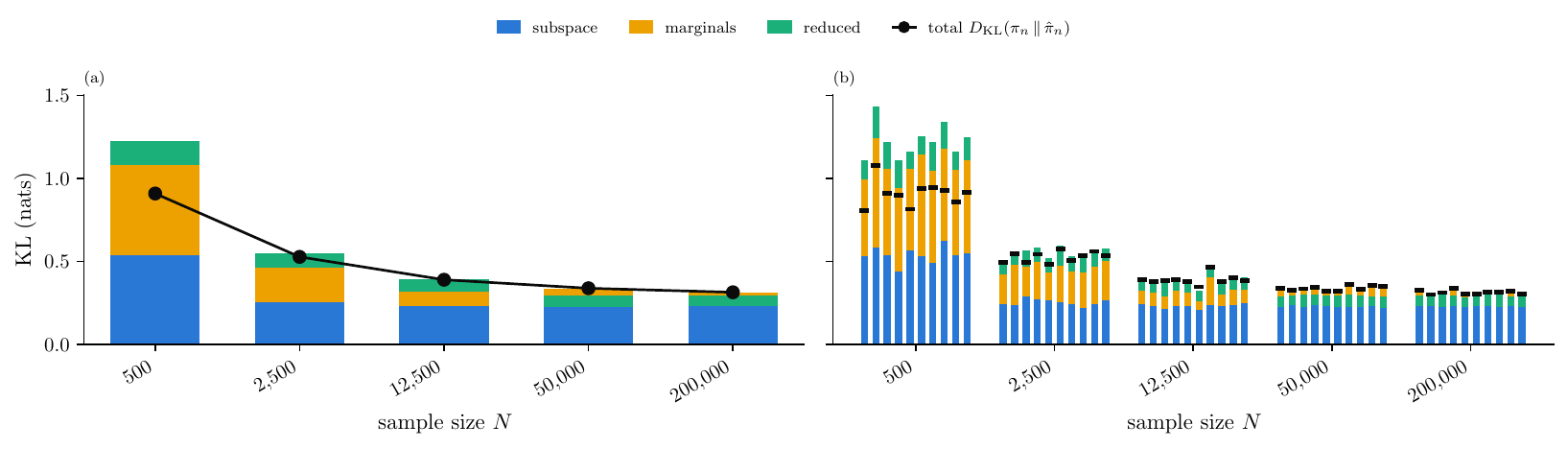}
\caption{Reference values of the three terms in \eqref{eq:kl-three-way} on Example~1 ($d=20$, $r=4$), for $10$ estimated models per $N$. See Supplement~\ref{sm:composite-anatomy} for Examples~2 and~3. \emph{(a)}~Mean over the $10$ models. \emph{(b)}~Individual models: each stacked bar has a tick at its reference total. We use the estimated basis of each model for the subspace term and $\mV_{r,\Lambda,\tau}$ for the reduced-density term. Neither $R_N$ nor $\log\widehat Z_{\mathrm{comp}}$ is included in the stacks (Supplement~\ref{sm:composite-anatomy}). For $N \ge 12{,}500$, the mean difference between stack and reference total is below the mean Monte Carlo standard error of an individual reference total.}
\label{fig:composite-main}
\end{figure*}

All three examples show a decreasing marginal term as $N$ grows. The subspace term is largest for $N \ge 12{,}500$, and its mean changes by less than $4\%$ from $N = 12{,}500$ to $N = 200{,}000$. Accuracy at these sample sizes is limited by the rank and multi-index set.

\subsection{A priori terms and the a posteriori bound}
\label{sec:exp-aposteriori}

In \Cref{fig:stage1-bounds}(a), we evaluate \eqref{eq:finite-sample-kl} for $N \in \{2500, 50000\}$ and $8$ seeds. We take $\scS'$ to be the Stage-1 estimate from $N_{\rm ref} = 10^6$ residuals at $(K_1,q_1) = (4,2)$. As in \eqref{eq:stage1-threeterm}, we split the trace into the rank-$r$ tail of $\mC'$ and the statistical term $\tr((\mP_r'-\hmP_r)\mC')$, with $\mP_r'$ the top-$r$ eigenprojector of $\mC'$. The Hermite term $\|\scS-\scS'\|_{L^2(\pi_{\vZ})}$, evaluated with the known copula score (Supplement~\ref{sm:experimental-protocol}), is largest on every example. The estimate $\scS'$ accounts for $24\%$, $26\%$ and $6\%$ of $\tr(\mC)$ on Examples~1, 2 and~3. These fractions differ by at most one percentage point for the unpenalized projection $\scS_\Lambda$ (Supplement~\ref{sm:experimental-protocol}). Only the statistical term depends on $N$, and it decreases as $N$ grows. After squaring the sum and dividing by two, the bound is between $16.9$ and $27.0$ times the reference $\mathcal E_{\rm proj}$ at $N = 5\times10^4$. The Hermite term contributes more than half of the bound. The contribution from the log-Sobolev inequality is assessed separately in \Cref{sec:tightness}.

In \Cref{fig:stage1-bounds}(b), the bound from \Cref{cor:aposteriori} is evaluated as $b_N$, using $\widehat m_4$ from the second half of each sample. An independent set of $2\times10^5$ residuals, mapped through the first-half rank transform, estimates the target trace $\tr((\mI-\hmP_r)\widetilde\mC(\whvtheta))$. The bound exceeds this trace for every seed and every $N$, a total of $96$ models. Its finite-sample term decreases with $N$.

For the same models, we evaluate \eqref{eq:aposteriori-corrected} with $b_N$ in place of $b_N^{\rm pop}$ and both corrections computed from the known score. Larger $N$ gives a smaller transform discrepancy $D_N$, and $\sqrt{d-r}\,D_N$ is at most $20\%$ of $b_N$. At $N \ge 12{,}500$, the score error $e_N$ lies no more than $10\%$ above the Hermite term of $\scS'$ and is the dominant term in \eqref{eq:aposteriori-corrected}. On all $96$ models, the corrected bound is greater than $C_1 = \tfrac12\tr((\mI-\hmP_r)\mC)$ in \eqref{eq:stage1-cert}, and therefore greater than $\mathcal E_{\rm proj}$. At $N = 5\times10^4$, its ratio to the reference $\mathcal E_{\rm proj}$ ranges from $21.8$ to $34.4$. The estimated-score trace in the complement directions is smaller than the true-score trace. At $N = 5\times10^4$, the corresponding half-trace $\tfrac12\tr((\mI-\hmP_r)\widetilde\mC(\whvtheta))$ falls below the reference $\mathcal E_{\rm proj}$ in Examples~1 and~3.

To estimate the Hermite part of the score error, \Cref{prop:hier-trunc} uses two nested multi-index sets under the saturation assumption \textup{(SAT)}. To compute $\beta$, we project the known score without penalization onto the two gradient spaces, enlarging $\Lambda$ first in degree and, separately, in interaction order. The degree increase gives $\beta$ above $0.98$ in every example, and the interaction increase gives $\beta$ between $0.62$ and $0.91$. Thus $1/(1-\beta^2)$, the ratio between the upper and lower bounds in \eqref{eq:hier-trunc}, exceeds $29$ for increased degree but is at most $6$ for increased interaction order.

\begin{figure}[t]
\centering
\includegraphics[width=\textwidth]{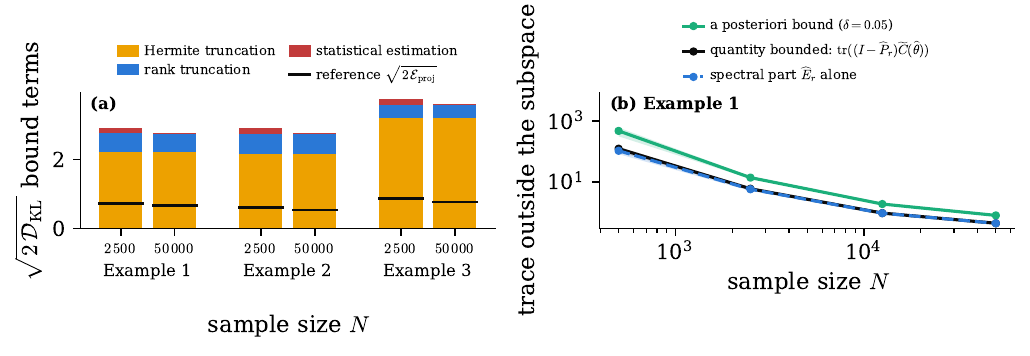}
\caption{Stage-1 error bounds for Examples~1--3 ($d=20$, $r=4$, $8$ seeds at each $N$). \emph{(a)}~The terms in \eqref{eq:finite-sample-kl} for $N \in \{2500, 50000\}$. The Stage-1 fit to $N_{\rm ref} = 10^6$ residuals supplies $\scS'$. The black line is the reference $\sqrt{2\,\mathcal E_{\rm proj}}$, and its distance below the top of the stack is the looseness of the bound in $\sqrt{2\,\DKL}$ units. \emph{(b)}~For Example~1 at $\delta = 0.05$, the bound $b_N$ of \Cref{cor:aposteriori} uses $\widehat m_4$ instead of $m_4$. Also shown are the target trace $\tr((\mI-\hmP_r)\widetilde\mC(\whvtheta))$ and the spectral term $\sum_{i>r}\widehat\lambda_i(\hmC)$. Shaded bands span the $8$ seeds. The difference between $b_N$ and the spectral term is the finite-sample term of \eqref{eq:aposteriori}, which decreases as $N$ grows. For $N \ge 12500$, the target traces on Examples~2 and~3 are within $34\%$ of that on Example~1.}
\label{fig:stage1-bounds}
\end{figure}

\subsection{The overall estimate against the reference values}
\label{sec:exp-endtoend}

We compare the two estimated terms of \eqref{eq:apost-assembled} with their reference values in \eqref{eq:kl-three-way}. Each example uses $N \in \{500, 2500, 12500, 50000, 200000\}$ and $10$ seeds, giving $150$ estimated models, all at the reference basis $\mV_{r,\Lambda,\tau}$ of \Cref{sec:exp-reference}. Each estimate uses $M = 2\times10^5$ held-out residuals mapped through that model's estimated marginal transform. We fix the number of neighbours at $k = 20$ for $r = 4$ and $k = 10$ for $r = 1$ \cite[\S2]{BerrettSamworthYuan2019} (\Cref{sec:limitations}).

For $\widehat{\mathcal E}_{\rm dens}$, the error bars use \Cref{cor:apost-den}, derived in exact coordinates with the exact normalizer. They omit the effects of the estimated transform, fixed $k$ and computed normalizer, all of which enter the comparison with reference values below. For $\widehat{\mathcal E}_{\rm marg}$, we sum the $d$ coordinate standard errors because the residual coordinates are dependent. Under the one-dimensional hypotheses of \Cref{cor:apost-den}, Minkowski's inequality makes this an estimate of an upper bound on the asymptotic standard error. These error bars also omit the variance contribution from fixing $k$.

\begin{figure*}[t]
\centering
\includegraphics[width=\textwidth]{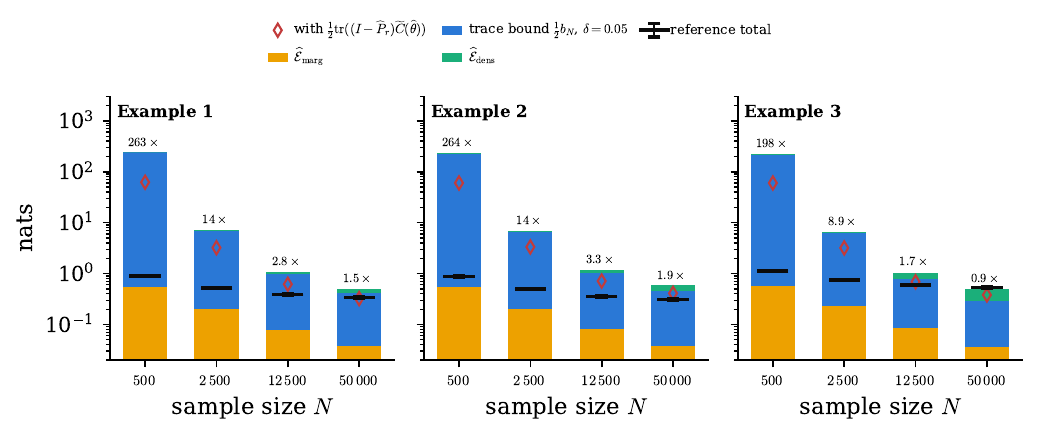}
\caption{The overall a posteriori estimate \eqref{eq:apost-assembled} against the reference total on Examples~1--3 ($d = 20$, $r = 4$, $\delta = 0.05$). Each stacked bar combines the means of the two estimated terms over $10$ models with the mean of $\tfrac12 b_N$ over the $8$ split-sample estimates in \Cref{sec:exp-aposteriori} at the same $N$. The black marker is the mean reference total $\DKL(\pi_{\vn}\,\|\,\widehat\pi_{\vn})$ for the $10$ models, with the Monte Carlo standard error of that mean. The red diamond replaces $\tfrac12 b_N$ by its target half-trace $\tfrac12\tr((\mI-\hmP_r)\widetilde\mC(\whvtheta))$. The gap between diamond and bar comes from the Chebyshev term in \eqref{eq:aposteriori}, with $\widehat m_4$ replacing $m_4$. The number above each bar is the effectivity index, the ratio of the estimate to the reference total. The vertical axis is logarithmic.}
\label{fig:apost-endtoend}
\end{figure*}

In all $150$ estimated models, $\widehat{\mathcal E}_{\rm marg}$ lies within twice the combined standard errors of the reference $\mathcal E_{\rm marg}$. The estimate $\widehat{\mathcal E}_{\rm dens}$ is biased upward. Its mean excess over the reference value is about one combined standard error for $N \ge 2{,}500$, and at most $2.4$ combined standard errors at $N = 500$, where the held-out residuals pass through a rank transform estimated from $500$ samples.

In \Cref{fig:apost-endtoend} the effectivity index decreases with $N$ on every example, from the hundreds at $N = 500$ to below $4$ for $N \ge 12{,}500$. At $N = 500$ the bounded term depends, through the fourth moment in \Cref{cor:aposteriori}, on a Stage-1 score estimated from $N_2 = 250$ samples. The mean estimate exceeds the mean reference total in $11$ of the $12$ combinations of example and sample size. The exception is Example~3 at the largest $N$ of the figure, where the index is $0.9$. There $\tfrac12 b_N$ is smaller than the reference $\mathcal E_{\rm proj}$, because the estimated score carries little of the score variance in the complement directions and \eqref{eq:apost-assembled} omits the score error $e_N$ (\Cref{sec:limitations}).

% =====================================================================
\section{Limitations}
\label{sec:limitations}

\emph{Unverifiable hypotheses.} Condition \textup{(D1)} of \Cref{cor:apost-den} requires the unknown reduced density to belong to the smoothness class of \cite{BerrettSamworthYuan2019}. That result concerns exact rank-Gaussianized coordinates. In \Cref{sec:exp-endtoend} the transform is estimated, and its image law has atoms. The corollary does not apply to it, and neither does the identity \eqref{eq:fitted-target}, which requires an image law with a density. A further difference is that we hold the number of neighbours $k$ fixed as $M$ increases. The reported standard error consequently omits part of the estimator's variance \cite[\S3.4]{BerrettSamworthYuan2019}. Residual samples do not suffice to check the smoothness assumptions or the transform condition \textup{(T)} required for the marginal rate. In \Cref{sec:exp-endtoend} we assess the practical estimates by comparison with known noise laws.

\emph{The score error.} The calculation of $b_N$ substitutes an estimated fourth moment for the population value. A bound on $\mathcal E_{\rm proj}$ still requires the score error $e_N$ in \eqref{eq:aposteriori-corrected}. On the examples, with $(K_1,q_1) = (4,2)$ and $10^6$ residuals, the trace of the fitted Stage-1 score's second-moment matrix is between $6\%$ and $26\%$ of $\tr(\mC)$. The score error $e_N$ is the dominant term in \eqref{eq:aposteriori-corrected} (\Cref{sec:exp-aposteriori}). The estimate of its Hermite component by \Cref{prop:hier-trunc} assumes \textup{(SAT)}, which two multi-index sets alone cannot verify. On Example~3, an increase of one in interaction order multiplies $\tr(\mC_\Lambda)$ by $3.7$, while an increase of one in degree raises it by only $0.5\%$. The hierarchical estimate therefore depends on the added indices as well as on the noise law.

\emph{Small samples.} At $N = 500$, only $250$ samples are available to estimate the Stage-1 score. The resulting bounded contribution exceeds the reference total by a factor greater than $100$ on every example (\Cref{sec:exp-endtoend}). At every sample size, the errors in the two estimated contributions are much smaller.

\emph{Scope.} The reduction represents dependence through a fixed linear subspace. When the directions of dependence vary along a curved set, the estimate represents dependence only through variation along the chosen subspace. Transport-based reductions \cite{MarzoukMoselhyParnoSpantini2016, ChenArnaudBaptistaZahm2024} offer alternatives for such laws. The number of multi-indices $|\Lambda_{K,q}|$ grows as $\binom{d}{q}$ with the interaction order $q$. Our analysis fixes the truncation and the regularizer in advance and does not cover choices that change with $N$. The reference values in \Cref{sec:experiments} all require a closed-form noise law. These experiments therefore do not assess the error estimate on residuals from an operational assimilation system.

% =====================================================================

\section{Conclusion}
\label{sec:conclusion}

When the noise law is unknown, we cannot compare the \textsc{Cas} density estimate with a reference. We instead estimate its error from residual samples. The exact decomposition in \Cref{prop:exact-decomp} separates the contributions of the marginals, subspace and reduced density. In exact rank coordinates, the marginal and reduced-density estimates have standard errors of order $M^{-1/2}$ in the number $M$ of held-out residuals under the assumptions of \Cref{cor:apost-den}. We test these estimates after the estimated transform in \Cref{sec:exp-endtoend}. The subspace bound combines a bound on the fitted-score trace outside the retained subspace with corrections for score estimation and the rank transform (\Cref{cor:aposteriori} and \eqref{eq:aposteriori-corrected}).

Across the $150$ estimated models, for which the noise laws are known, the marginal estimate agrees with its reference value, and the reduced-density estimate has a small upward bias. After averaging over the estimated models, \eqref{eq:apost-assembled} exceeds the reference divergence in $11$ of the $12$ combinations of example and sample size. The exception is discussed in \Cref{sec:exp-endtoend}.

The Gaussian log-Sobolev bound is between $1.64$ and $1.98$ times the divergence on the three examples. Numerical descent of the tighter bound of \cite{LiCuiLiMarzoukZahm2024sharp} ends at a basis within $\|\sin\Theta\|_F \le 0.0004$ of the top-$r$ eigenbasis of $\mC$ used by \textsc{Cas} in each example (\Cref{sec:tightness}). A larger discrepancy comes from truncating the score to the chosen multi-index set. The Hermite term contributes more than half of the a priori bound on every example. The same truncation also contributes to the score error $e_N$ in the a posteriori bound.

For $N \ge 12{,}500$, the subspace term is the largest error contribution on each example and changes little with further sampling. The rank and multi-index truncation limit accuracy in this regime. The Stage-1 spectrum gives the sum of trailing eigenvalues for every rank. Under the saturation assumption, the increment between two nested multi-index sets estimates the Hermite truncation term (\Cref{prop:hier-trunc}). In the examples, enlarging the interaction order gives a saturation constant well below one (\Cref{sec:exp-aposteriori}). Comparing these two quantities helps determine whether to increase the rank or enlarge the multi-index set.

\section*{Reproducibility}
\label{sec:reproducibility}

Code, saved numerical results, and reproduction instructions are available at \url{https://github.com/joshuawchen/copula_active_subspaces}. This article uses the release tagged \texttt{v1.0-part2}. The script \texttt{reproduce.py} generates the figures from saved numerical results and recomputes these results on request. The repository README lists the scripts used for each figure and table. Supplement~\ref{sm:experimental-protocol} records the settings and random states of the reported runs.

\appendix
\crefalias{section}{appendix}
\crefalias{subsection}{appendix}

% =====================================================================
\section{Proofs of the decomposition and the Stage-1 bounds}
\label{app:thm-trunc}

\begin{proof}[Proof of \Cref{prop:stage1-cert}, \Cref{thm:stage1-rate}, and \eqref{eq:aposteriori-kl}]
Apply the Gaussian-reference bound \textup{(2.5)} to the basis $\hmV_r^{\mC}$:
\[
  2\,\DKL(\pi_{\vn}\,\|\,\pi_{\vn}(\,\cdot\,;\hmV_r^{\mC}))
  \;\le\;
  \tr\!\bigl((\mI - \hmP_r)\mC\bigr)
  \;=\;
  \|(\mI - \hmP_r)\scS\|_{L^2(\pi_{\vZ})}^2,
\]
where the equality uses $\mC = \EE_{\pi_{\vZ}}[\scS\scS^\top]$ and the idempotence of $\mI-\hmP_r$. Division by two proves \eqref{eq:stage1-cert}. For any score field $\scS'$ with second-moment matrix $\mC' := \EE_{\pi_{\vZ}}[\scS'\scS'^\top]$, apply the triangle inequality in $L^2(\pi_{\vZ};\RR^d)$. The orthogonal projector $\mI-\hmP_r$ does not increase norms, and $\|(\mI-\hmP_r)\scS'\|_{L^2(\pi_{\vZ})}^2 = \tr((\mI-\hmP_r)\mC')$. Hence $\|\scS - \scS'\|_{L^2} + \sqrt{\tr((\mI-\hmP_r)\mC')}$ bounds $\|(\mI-\hmP_r)\scS\|_{L^2}$, proving \eqref{eq:finite-sample-kl} as in \cite[Thm.~3.3]{BaptistaBrennanMarzouk2025}. The choice $\scS' = \widehat\scS$ proves \eqref{eq:aposteriori-kl} of \Cref{cor:aposteriori}.

By Part~I's \PartIref{corollary}{2.3}, $\tr((\mI-\mP_r^{\mC})\mC) = \sum_{i>r}\lambda_i(\mC)$. Subtracting this minimum from the trace proves the split \eqref{eq:stage1-split} and the nonnegativity $\Delta_{\rm subopt} \ge 0$. Apply the eigengap inequality of Part~I (\PartIref{proposition}{4.1}, which is \cite[Lem.~4.2]{VuLei2013}) with $\mC$ in place of $\mC_\Lambda$. It gives uniqueness of the maximizer when $\lambda_r(\mC) > \lambda_{r+1}(\mC)$, as well as $\Delta_{\rm subopt} = O(\|\sin\Theta(\hmV_r^{\mC},\mV_r^{\mC})\|_F^2)$. For the truncated score, Part~I's \PartIref{proposition}{4.1} states that $\mP_r^{\Lambda\star}$ minimizes $\tr((\mI-\mP)\mC_\Lambda)$ over rank-$r$ orthogonal projectors. Its minimum is $\sum_{i>r}\lambda_i(\mC_\Lambda)$, so
\begin{equation}
  \tr\!\bigl((\mI-\hmP_r)\mC_\Lambda\bigr) \;=\; \textstyle\sum_{i>r}\lambda_i(\mC_\Lambda) \;+\; \tr\!\bigl((\mP_r^{\Lambda\star}-\hmP_r)\mC_\Lambda\bigr),
  \label{eq:stage1-decomp}
\end{equation}
where the second summand is nonnegative. Apply $\sqrt{a+b}\le\sqrt a+\sqrt b$ to \eqref{eq:stage1-decomp} and use \eqref{eq:finite-sample-kl} with $\scS' = \scS_\Lambda$ to obtain \eqref{eq:stage1-threeterm}. For \eqref{eq:stage1-estrate}, Weyl's inequality \cite[Cor.~III.2.6]{Bhatia1997}, the Ky Fan trace extremum \cite[Prob.~III.6.11]{Bhatia1997}, and a rank-$r$ trace bound give $\tr((\mP_r^{\Lambda\star}-\hmP_r)\mC_\Lambda) \le 2r\|\hmC-\mC_\Lambda\|_{\rm op}$. Equation \eqref{eq:opnorm-rate} of \Cref{lem:score-cov-rate} bounds this quantity by $O_p(N^{-1/2}) + O(\tau_N)$ without an eigengap assumption. Assumption \textup{(A4)} enters only the $\sin\Theta$ rate \eqref{eq:sintheta-rate}, which is tested in Supplement~\ref{sm:rate-real}.
\end{proof}

\begin{proof}[Proof of \eqref{eq:aposteriori} in \Cref{cor:aposteriori}]
Conditional on the first half, the rank transform and $\whvtheta$ are fixed. The scores $\widehat\scS(\vZ^{(n)})$, $n = 1, \dots, N_2$, from the second half are i.i.d.\ and satisfy $\EE[\widehat\scS\widehat\scS^\top] = \widetilde\mC(\whvtheta)$. Since $\sum_{i>r}\widehat\lambda_i(\hmC) = \tr((\mI-\hmP_r)\hmC)$,
\[
  \tr\!\bigl((\mI-\hmP_r)\widetilde\mC(\whvtheta)\bigr) \;=\; \textstyle\sum_{i>r}\widehat\lambda_i(\hmC) + \tr\!\bigl((\mI-\hmP_r)(\widetilde\mC(\whvtheta) - \hmC)\bigr).
\]
Because the orthogonal projector $\mI-\hmP_r$ has rank $d-r$, von Neumann's trace inequality \cite{Mirsky1975} bounds the second term by the sum of the $d-r$ largest singular values of $\widetilde\mC(\whvtheta)-\hmC$. Cauchy--Schwarz bounds that sum by $\sqrt{d-r}\,\|\widetilde\mC(\whvtheta)-\hmC\|_F$. These inequalities apply to each realized $\hmP_r$. The centered summands $\widehat\scS(\vZ^{(n)})\widehat\scS(\vZ^{(n)})^\top - \widetilde\mC(\whvtheta)$ are i.i.d.\ and mean-zero, so
\[
  \EE[\|\hmC - \widetilde\mC(\whvtheta)\|_F^2] \;=\; N_2^{-1}\bigl(\EE[\|\widehat\scS(\vZ)\|_2^4] - \|\widetilde\mC(\whvtheta)\|_F^2\bigr) \;\le\; m_4/N_2,
\]
and Chebyshev's inequality applied to $\|\hmC-\widetilde\mC(\whvtheta)\|_F^2$ gives $\|\hmC-\widetilde\mC(\whvtheta)\|_F \le \sqrt{m_4/(N_2\delta)}$ with probability at least $1-\delta$. This proves \eqref{eq:aposteriori}. The moment $m_4$ is finite for every noise law, because the first-half rank transform has bounded range and the fixed polynomial $\widehat\scS$ is bounded on it.

To replace the moment by $\widehat m_4$, condition again on the first half and first suppose $m_4>0$. Since $\Var(\widehat m_4/m_4)=(m_8/m_4^2-1)/N_2$, the moment-ratio condition implies $\widehat m_4/m_4=1+O_p(N_2^{-1/2})$ over both halves. Let $\epsilon\in(0,1)$. On the event $\widehat m_4\ge(1-\epsilon)m_4$, failure of \eqref{eq:aposteriori} with $\widehat m_4$ implies $\|\hmC-\widetilde\mC(\whvtheta)\|_F>\sqrt{(1-\epsilon)m_4/(N_2\delta)}$. The preceding bound puts this conditional probability at most $\delta/(1-\epsilon)$. With $\epsilon=N_2^{-1/4}$, Chebyshev's inequality bounds the conditional probability of the complementary event by $(m_8/m_4^2-1)N_2^{-1/2}$. Adding the two probabilities and averaging over the first half bounds the failure probability by $\delta+o(1)$. If $m_4=0$, the second-half scores vanish almost surely. Then $\hmC$, $\widetilde\mC(\whvtheta)$ and $\widehat m_4$ vanish, and both sides of \eqref{eq:aposteriori} are zero.
\end{proof}

\begin{proof}[Proof of \Cref{prop:hier-trunc}]
The Pythagorean identity follows because $\scS-\scS_{\Lambda'}$ is orthogonal to $\scF_{\Lambda'}$, and \textup{(SAT)} then proves \eqref{eq:hier-trunc}. For the sample estimate, embed both coefficient vectors in the larger dictionary and write $\Delta\widehat\theta_j=\Delta\theta+e_j$, with $e_j=O_p(N^{-1/2})+O(\tau_N)$, and $\widehat\mA_{\Lambda'}=\mA_{\Lambda'}+E_A$, with $E_A=O_p(N^{-1/2})$. Expansion of the bilinear form around $\Delta\theta^\top\mA_{\Lambda'}\Delta\theta=\eta^2$ yields the stated order. The estimated-score bound follows by the triangle inequality.
\end{proof}

\begin{proof}[Proof of \Cref{prop:exact-decomp} and of the $R_N$ rate of \Cref{prop:apriori-rates}]
The exact rank-transform change of variables and the estimator definition in \Cref{sec:estimator} give $\pi_{\vn}(\vx) = c^{Z}(\scZ(\vx))\prod_i\pi_i(x_i)$ and $\widehat Z_{\mathrm{comp}}\,\widehat\pi_{\vn}(\vx) = \widehat c^{Z}(\widehat\scZ(\vx))\prod_i\widehat\pi_i(x_i)$, with $0 < \widehat Z_{\mathrm{comp}} < \infty$ (Part~I \cite{CAS_partI}). To justify integration, check each logarithm under $\pi_{\vn}$. The function $|\log c^{Z}| = |\scL|$ grows polynomially by \textup{(T2)} (Part~I \cite{CAS_partI}), while the Gaussian marginals of $\pi_{\vZ}$ supply all polynomial moments. The fitted $\log\widehat c^{Z}$ is a degree-$K_2$ Stage-2 polynomial in $\hmV_r^{\mC\top}\vz$, minus the log-normalizer. It is bounded above by the Stage-2 constraint (Part~I \cite{CAS_partI}) and grows at most polynomially below. For the marginal logarithms, $\EE_{\pi_{\vn}}[|\log(\pi_i/\widehat\pi_i)|] < \infty$ follows from \textup{(M)} (\Cref{sec:estimator}).

\emph{The exact identity.} Pointwise, $\log(\pi_{\vn}/\widehat\pi_{\vn}) = [\log c^{Z}(\scZ) - \log\widehat c^{Z}(\widehat\scZ)] + \sum_i\log(\pi_i/\widehat\pi_i) + \log\widehat Z_{\mathrm{comp}}$. Since $x_i \sim \pi_i$ under $\pi_{\vn}$, taking expectations gives the marginal contribution $\sum_i\DKL(\pi_i\,\|\,\widehat\pi_i)$, and the constant term remains $\log\widehat Z_{\mathrm{comp}}$. Split the copula contribution by adding and subtracting $\EE_{\pi_{\vn}}[\log\widehat c^{Z}(\scZ)]$. The transform $\scZ$ pushes $\pi_{\vn}$ forward to $\pi_{\vZ}$, so the first difference equals $\EE_{\pi_{\vZ}}[\log(\pi_{\vZ}/\gauss_d) - \log(\widehat\pi_{\vZ}/\gauss_d)] = \DKL(\pi_{\vZ}\,\|\,\widehat\pi_{\vZ})$. Here $\widehat\pi_{\vZ} := \widehat c^{Z}\gauss_d$ is normalized by $\widehat Z_r$ (\Cref{sec:apost-density}). The second difference is $R_N$. Conditional on the estimation sample, the factorization $\widehat\pi_{\vZ}(\vz) = \widehat c_r^{U}(\vu)\gauss_r(\vu)\,\gauss_{d-r}(\vw)$ holds in the coordinates $(\vu,\vw) = (\hmV_r^{\mC\top}\vz,\ \hmV_\perp^\top\vz)$. The KL chain rule in these coordinates gives
\[
  \DKL(\pi_{\vZ}\,\|\,\widehat\pi_{\vZ}) \;=\; \DKL(\pi_{\vU}\,\|\,\widehat\pi_{\vU}) + \EE_{\pi_{\vU}}[\DKL(\pi_{\vW\mid\vU}\,\|\,\gauss_{d-r})],
\]
where the second term is $\DKL(\pi_{\vn}\,\|\,\pi_{\vn}(\,\cdot\,;\hmV_r^{\mC}))$ by Part~I's reduced-form identity \cite{CAS_partI}. Because the identity holds for any fixed subspace, it holds at each realized $\hmV_r^{\mC}$. Every term is finite by the integrability checks above.

\emph{The conditional rate of $R_N$.} For the fixed finite dictionary, the polynomial mean-value formula and $\|\hmV_r^{\mC\top}\|_{\rm op}=1$ give a deterministic constant $C_\Lambda$ such that
\[
|\widehat\scL_r(\hmV_r^{\mC\top}\vz)-\widehat\scL_r(\hmV_r^{\mC\top}\vz')|
\le C_\Lambda\|\hvbeta\|\,(1+\|\vz\|+\|\vz'\|)^{K_2-1}\|\vz-\vz'\|.
\]
With $W_N$ denoting the inner expectation in \eqref{eq:transform-moment}, Cauchy--Schwarz over the estimation sample and \textup{(T)} yield the following bound. No independence between $\hvbeta$ and $\widehat\scZ$ is needed:
\[
\begin{aligned}
\EE|R_N|&\le C_\Lambda\EE[\|\hvbeta\|W_N^{1/2}]
\le C_\Lambda(\EE\|\hvbeta\|^2)^{1/2}(\EE W_N)^{1/2}\\
&=O(\sqrt d\,(\log N)^{K_2/2}N^{-1/2}).
\end{aligned}
\]
Combining this bound with \textup{(M)} proves \eqref{eq:Delta-marg}, and Markov's inequality gives the $O_p$ statement. Beyond the sample range, the estimated distribution functions of Part~I take the values $1/[2(N+1)]$ and $1-1/[2(N+1)]$. Thus $\|\widehat\scZ\|_\infty\le z_N:=\Phi^{-1}(1-1/[2(N+1)])$ and $\|\hmV_r^{\mC\top}\widehat\scZ\|\le\sqrt d\,z_N$, with $z_N^2\le2\log(N+1)$.

\emph{Expectation and assembly.} Take expectations over the estimation sample in the exact identity to obtain $\EE[\DKL(\pi_{\vn}\,\|\,\widehat\pi_{\vn})] = \EE[\mathcal E_{\rm marg}] + \EE[\mathcal E_{\rm proj}] + \EE[\mathcal E_{\rm dens}]$. Since \Cref{prop:stage1-cert} holds for every realization, $\EE[\mathcal E_{\rm proj}] \le \EE[C_1]$. Now condition on the sample used to compute $\hmV_r^{\mC}$. The Stage-2 sample is independent of this sample and is projected onto a fixed basis. Applying \eqref{eq:stage2-kl-mean} at that basis gives $\EE[\mathcal E_{\rm dens} \mid \hmV_r^{\mC}] \le \mathcal B_{\rm dens}(\hmV_r^{\mC}) + O(N^{-1})$. Uniformity of the assumed constants makes the remainder uniform over the basis. The expectation of this conditional bound, together with the marginal and projection bounds, proves \eqref{eq:composite-kl}.
\end{proof}

\section{The Stage-1 score-covariance rate}
\label{app:sample-proof}

The conditions of \Cref{lem:score-cov-rate} are as follows. \textup{(A1)} $\Lambda$ is finite, excludes the constant mode, and $\pi_{\vZ}>0$ a.e.\ \textup{(A2)} The empirical Gram matrix, right-hand side and every gradient cross-moment $N^{-1}\sum_n\nabla H_\alpha(\vZ^{(n)})\nabla H_\beta(\vZ^{(n)})^\top$ differ from their exact-coordinate expectations by $O_p(N^{-1/2})$. \textup{(A3)} Stage~1 uses $\tau_N\mD$ with bounded deterministic $\tau_N\ge0$ and fixed positive diagonal $\mD$. \textup{(A4)} The gap $\delta_\Lambda:=\lambda_r(\mC_\Lambda)-\lambda_{r+1}(\mC_\Lambda)$ is positive.

For i.i.d.\ exact rank coordinates, \textup{(A2)} follows from the multivariate central limit theorem. For empirical ranks, \textup{(A2)} is an assumption, since central limit theorems for rank statistics require growth conditions on the score functions \cite[Appendix~1]{GenestGhoudiRivest1995}. The implementation uses $\widehat\tau=\kappa\|\widehat\mA\|_{\rm op}$ with $\kappa=10^{-4}$. By \textup{(A2)}, its difference from $\tau_0=\kappa\|\mA^{d,\star}\|_{\rm op}$ is $O_p(N^{-1/2})$, so this choice contributes an additional $O_p(N^{-1/2})$ term in the proof below.

\begin{lemma}[Stage-1 score-covariance rate]
\label{lem:score-cov-rate}
Under Assumption~2.1 of Part~I and the conditions \textup{(A1)}--\textup{(A3)} above, the empirical Stage-1 score second-moment matrix $\hmC$ of \textup{(3.5)} satisfies
\begin{equation}
  \|\hmC - \mC_\Lambda\|_{\rm op} \;=\; O_p(N^{-1/2}) \;+\; O(\tau_N)
  \label{eq:opnorm-rate}
\end{equation}
at the Tikhonov regularization strength $\tau_N$ of \textup{(A3)}. If the eigengap condition \textup{(A4)} holds as well, then the estimated Stage-1 basis satisfies
\begin{equation}
  \|\sin\Theta(\hmV_r^{\mC}, \mV_{r,\Lambda})\|_F \;=\; O_p(N^{-1/2}) \;+\; O(\tau_N/\delta_\Lambda).
  \label{eq:sintheta-rate}
\end{equation}
\end{lemma}

\begin{proof}[Proof of \Cref{lem:score-cov-rate}]
To see that $\mA^{d,\star}$ is positive definite under \textup{(A1)}, consider a linear combination of the included gradients that vanishes a.e.\ It is polynomial, hence identically zero. The corresponding polynomial is therefore constant, and since the constant mode is excluded, all its coefficients vanish. Define
\[
\theta_{\tau_N}=(\mA^{d,\star}+\tau_N\mD)^{-1}\mb^{d,\star},\qquad
\theta_\Lambda=(\mA^{d,\star})^{-1}\mb^{d,\star}.
\]
Near $\mA^{d,\star}$, the inverse is uniformly bounded for bounded nonnegative $\tau_N$. The resolvent identity and \textup{(A2)} give $\widehat\theta-\theta_{\tau_N}=O_p(N^{-1/2})$. The regularization error satisfies $\theta_{\tau_N}-\theta_\Lambda=-\tau_N(\mA^{d,\star}+\tau_N\mD)^{-1}\mD\theta_\Lambda=O(\tau_N)$. The empirical second-moment matrix is quadratic in the coefficients, with the finitely many gradient cross-moments in \textup{(A2)} as weights. Substituting the coefficient and cross-moment errors proves \eqref{eq:opnorm-rate}. Using the population-gap version of Davis--Kahan \cite[Thm.~2]{YuWangSamworth2015}, the projector error is bounded by a constant times $\sqrt r\,\|\hmC-\mC_\Lambda\|_{\rm op}/\delta_\Lambda$. Equation \eqref{eq:sintheta-rate} follows, with constants depending on the fixed population gap. No bound on the empirical eigengap is required. For $\widehat\tau=\tau_0+O_p(N^{-1/2})$, the resolvent identity gives the same coefficient rate around $\theta_{\tau_0}$. Its second-moment matrix is the regularized limit $\mC_{\Lambda,\tau_0}$.
\end{proof}

% =====================================================================
\section{Proof of the a posteriori divergence estimate}
\label{app:apost-den}

\begin{proof}[Proof of \Cref{cor:apost-den}]
Under the density, weight and neighbour conditions in the corollary, the mean-square equivalence of \cite[Thm.~1]{BerrettSamworthYuan2019} gives
$\widehat H^w_M-H(\pi_{\vU})=M^{-1}\sum_m[-\log\pi_{\vU}(\vu_m)-H(\pi_{\vU})]+o_p(M^{-1/2})$.
Assumption \textup{(D3)} makes the cross-entropy term in \eqref{eq:apost-est} an i.i.d.\ sample mean. Subtracting the entropy expansion from that mean gives
\[
  \widehat{\mathcal E}_{\rm dens} - \mathcal E_{\rm dens}
  \;=\;
  \frac1M\sum_{m=1}^M\Bigl[\log\frac{\pi_{\vU}(\vu_m)}{\widehat\pi_{\vU}(\vu_m)} - \mathcal E_{\rm dens}\Bigr] \;+\; o_p(M^{-1/2}),
\]
where the i.i.d.\ summands have mean zero and variance $\sigma_{\rm den}^2 < \infty$ by \textup{(D2)}, so the central limit theorem and Slutsky's lemma prove \eqref{eq:apost-clt}. For consistency, combine \cite[Thm.~1]{BerrettSamworthYuan2019} with the law of large numbers for the cross-entropy. Apply Minkowski's inequality to the centered log densities in $L^2(\pi_{\vU})$ to bound the variance of $\log\pi_{\vU} - \log\widehat\pi_{\vU}$. The first variance term in that bound is estimated as in \cite[\S2]{BerrettSamworthYuan2019}. For $r = 1$, the weight class has no $\Gamma$-constraints, and $w = e_k$ gives the unweighted estimator.
\end{proof}

\subsection{Quadrature of the reduced normalizer}
\label{app:qmc-normalizer}

Fix the fitted polynomial $p:=\widehat\scL_r$ and the dimension $r$.
Suppose its degree $K\ge2$ is even and its degree-$K$ homogeneous part
satisfies $F(\vv)\le-\varepsilon$ for all $\|\vv\|=1$, with $\varepsilon>0$.
The leading term dominates the lower-degree terms, yielding
$p(\vu)\le a-b\|\vu\|^K$ for some $a\in\RR$ and $b>0$.
Define $g(\bm t):=\exp(p(\Phi^{-1}(\bm t)))$ on $(0,1)^r$, with the inverse
normal CDF applied componentwise, and extend $g$ by zero to the boundary.
The extension is continuous by the tail bound. Writing $\phi$ for the standard
normal density and differentiating once in each coordinate gives
\[
 \frac{\partial^r g}{\partial t_1\cdots\partial t_r}(\bm t)
 =\frac{e^{p(\vu)}P(\vu)}{\prod_{j=1}^r\phi(u_j)},
 \qquad \vu=\Phi^{-1}(\bm t),
\]
where $P$ is a polynomial. Changing variables therefore gives
\[
 \int_{(0,1)^r}\left|\frac{\partial^r g}
 {\partial t_1\cdots\partial t_r}\right|\,d\bm t
 =\int_{\RR^r}e^{p(\vu)}|P(\vu)|\,d\vu<\infty.
\]
On any interior grid partition, the fundamental theorem of calculus bounds
the sum of absolute rectangular increments by this integral. The same bound
holds for partitions meeting the boundary, by continuity. Every proper face
has zero variation because $g$ vanishes there. Thus the Hardy--Krause variation
$V_{\rm HK}(g)$ is finite.

Let $m=M_{\rm norm}$, and let
$\widehat Z_{r,m}^{\rm QMC}:=m^{-1}\sum_{j=1}^m g(\bm t_j)$ for a point set
with star discrepancy $D_m^*$.
The Koksma--Hlawka inequality \cite[Thm.~2.11]{Niederreiter1992} gives
\[
 |\widehat Z_{r,m}^{\rm QMC}-\widehat Z_r|
 \le V_{\rm HK}(g)D_m^*.
\]
The first $m=2^\ell$ points of a Sobol' sequence \cite{Sobol1967} form a
$(t,\ell,r)$-net in base $2$. The random linear scrambling with a digital shift
of the implementation \cite{Matousek1998} maps such a net to a $(t,\ell,r)$-net,
so for fixed $r$ the star discrepancy is $O(m^{-1}(\log m)^{r-1})$
\cite[Thm.~4.10]{Niederreiter1992}.
Once $m$ is large enough, the quadrature error is at most $\widehat Z_r/2$.
The mean-value theorem for the logarithm then gives
\[
 |\log\widehat Z_{r,m}^{\rm QMC}-\log\widehat Z_r|
 \le \frac{2V_{\rm HK}(g)}{\widehat Z_r}D_m^*
 =O(m^{-1+\epsilon})\qquad(\epsilon>0).
\]
The bound holds for a fixed fitted polynomial evaluated without modification.
Clipping the quadrature coordinates or capping polynomial values introduces
numerical errors that this argument does not bound.

% =====================================================================
\section{Stage-2 regularizer: proofs and constant selection}
\label{app:bisc-ridge}

The Stage-2 objective is a finite-dimensional score-matching quadratic with Tikhonov regularization \cite{Sriperumbudur2017,EnglHankeNeubauer1996}. We use the eigenbasis calculation for ridge regression \cite{HoerlKennard1970} to choose the form of the penalty, then prove an error bound for the resulting data-dependent regularizer. Write $\Lambda^{(k)}:=\{\alpha\in\Lambda_r:|\alpha|_1=k\}$ and $p_k:=|\Lambda^{(k)}|$. Let $\widehat\mM:=\widehat\mA_r+\widehat\mR$, where $\widehat\mR=\mR(\widehat\mA_r)$ is the regularizer of Part~I. Under the integration-by-parts conditions of Part~I, $\EE[b_\alpha(\vU)]=\EE[\nabla H_\alpha(\vU)\cdot\nabla\scL_r(\vU)]$. The population normal equations are therefore those of the Fisher projection.

\paragraph{Standing hypotheses} The proofs in this appendix use the following structural conditions on the Stage-2 problem.
\begin{itemize}
\item[\textbf{(H1)}] \emph{Bounded reduced factor.} Assume $C_\star:=\mathrm{ess\,sup}_{\RR^r}(\pi_{\vU}/\gauss_r)<\infty$. Part~I's regularity assumptions do not imply this condition, and it must be checked when the rate theorem is applied. Under this condition, the score-matching integrands for a fixed dictionary have moments of every order.
\item[\textbf{(H2)}] \emph{Sobolev regularity of the exact reduced log-copula.} The function $\scL_r = \log c_r^{U}$ lies in $H^s_{\gauss}(\RR^r)$ for some $s > 1/2$. Write $\vartheta_\alpha := \langle\scL_r, H_\alpha\rangle_{L^2(\gauss_r)}$ for its Hermite coefficients, indexed by $\alpha \in \NN_0^r$. The assumption is $\|\scL_r\|_{H^s_{\gauss}}^2 = \sum_\alpha |\alpha|_1^{2s}\vartheta_\alpha^2 =: M^2 < \infty$. It implies the per-mode coefficient bound \eqref{eq:Hs-norm}.
\end{itemize}

\smallskip
\noindent Throughout the analysis, the basis $\mV_r$ and the dictionary are fixed, and the observations from $\pi_{\vU}$ are independent and use the exact rank transform. In \Cref{prop:apriori-rates}, the same analysis is conditional on a basis computed from an independent sample.

\subsection{The Stage-2 regularizer}
\label{sec:ridge}

Part~I constructs the Stage-2 regularizer $\mR$ and selects its constants by validation score-matching risk (\PartIref{subsection}{3.3}). The penalty is diagonal. It has a sampling term, whose scale decreases with sample size within each degree class, and a curvature term weighted by degree. These terms stabilize the empirical normal equations and add a regularization bias. \Cref{thm:ccsm-rate} bounds this bias together with the sampling error.

\paragraph{Form}
Part~I defines the regularizer $\mR(\widehat\mA_r) = \mR_{\rm samp}(\widehat\mA_r) + \mR_{\rm curv}(\widehat\mA_r)$ in \textup{(3.7)}. It is constant on each degree class $\Lambda^{(k)}$, with entry $\nu^{(k)} := C_{\rm samp}\widehat\sigma_k^2 + C_{\rm curv}\|\widehat\mA_r\|_{\rm op}k^2$. Here $\widehat\sigma_k^2 := \tfrac{1}{N}\tr(\widehat\mA_r^{(k,k)})$ is the degree-class variance scale. As $N$ increases, $\widehat\sigma_k^2 \asymp k\,p_k/N$ decreases to $0$, whereas $\|\widehat\mA_r\|_{\rm op}$ has an $N$-independent limit. The two regimes in \Cref{thm:ccsm-rate} correspond to these different scales. The exponent $2$ in the $|\alpha|_1^2$ weighting is fixed, and neither the theory nor the choice of constants adapts it to the unknown $s$ in \textup{(H2)}.

For the per-mode optimal regularizer $\nu_i^\star$ in the deterministic-Gram calculation of \Cref{app:bias-var}, we substitute the per-class noise bound of \Cref{cor:per-class-cov-bound} for the variance and the Sobolev decay \eqref{eq:Hs-norm} for the signal. The decay is known in the canonical Hermite basis, so applying it in the eigenbasis of $\mA^\star$ is an approximation used to design the penalty. The bound in \Cref{thm:ccsm-rate} is proved directly in the canonical basis.

Part~I also imposes the centering constraint $\widehat\mB\hvbeta = 0$. The empirical centering matrix $\widehat\mB$ has entries $\widehat B_{j\alpha} := \tfrac{1}{N}\sum_n \partial_j H_\alpha(\vu^{(n)})$. \PartIref{proposition}{3.1} gives its finite-sample effect.

\paragraph{Stage-2 estimator}
The Stage-2 quadratic is minimized over the closed convex constraint set of Part~I, which includes empirical centering and is assumed nonempty. Its distance from $\bm\theta^{r,\star}$ appears explicitly in the error bound of \Cref{thm:ccsm-rate}. The solver approximates this set by finitely many cuts, as described in Part~I.

\emph{The constants.} Part~I selects $(C_{\rm samp},C_{\rm curv})$ on a $5\times5$ grid $\mathcal P$ by validation score-matching risk. For the estimate $\hvbeta_{\mathbf c}$ at a pair $\mathbf c$, we have $\mathcal J_{\rm est}(\hvbeta_{\widehat{\mathbf c}})\le\sum_{\mathbf c\in\mathcal P}\mathcal J_{\rm est}(\hvbeta_{\mathbf c})$ for every realization, regardless of the selection rule. Taking expectations bounds the error at the selected pair by the sum of the bounds in \Cref{thm:ccsm-rate}, as long as its hypotheses hold at every pair in $\mathcal P$.

\subsection{Truncation and estimation error}
\label{app:bias-var}

Throughout this appendix the Fisher divergence $\mathcal J(\widehat\scL_r) := \|\nabla(\widehat\scL_r - \scL_r)\|^2_{L^2(\pi_{\vU})}$ satisfies $\mathcal J(\widehat\scL_r) = 2J(\hvbeta)$, where $J$ is the halved Fisher divergence of Part~I, eq.~\textup{(3.6)}. The factor of $2$ is included in the constants of \Cref{thm:ccsm-rate}.

For the estimated $\widehat\scL_r = \sum_{\alpha\in\Lambda_r}\widehat\theta^r_\alpha H_\alpha$, the triangle inequality in $L^2(\pi_{\vU};\RR^r)$ bounds the Fisher divergence by truncation and estimation terms:
\begin{equation}
  \sqrt{\mathcal J(\widehat\scL_r)} \;\le\; \underbrace{\bigl\|\nabla(\scL_r - \scL_r^\star)\bigr\|_{L^2(\pi_{\vU})}}_{\sqrt{\mathcal J_{\rm trunc}}} \;+\; \underbrace{\bigl\|\nabla(\scL_r^\star - \widehat\scL_r)\bigr\|_{L^2(\pi_{\vU})}}_{\sqrt{\mathcal J_{\rm est}(\hvbeta)}},
  \label{eq:triangle}
\end{equation}
Here $\scL_r^\star := \sum_{\alpha\in\Lambda_r}\theta^{r,\star}_\alpha H_\alpha$ is the infinite-sample Stage-2 minimizer over $\Lambda_r$. Its gradient is the $L^2(\pi_{\vU})$-projection of $\nabla\scL_r$ onto $\mathrm{span}\{\nabla H_\alpha : \alpha\in\Lambda_r\}$. The coefficients $\theta^{r,\star}_\alpha$ coincide with the Hermite coefficients $\vartheta_\alpha$ of \textup{(H2)} when $\pi_{\vU}=\gauss_r$, and also in other exactly represented cases.

Under \textup{(H1)}--\textup{(H2)} with $s>1/2$, we have $\mathcal J_{\rm trunc} \le C_\star\,K^{-(2s-1)}\|\scL_r\|^2_{H^s_{\gauss}}$, where $K := \min\{|\alpha|_1\ge1 : \alpha\notin\Lambda_r,\ \vartheta_\alpha \ne 0\}$. Use $\sum_{\alpha\in\Lambda_r}\vartheta_\alpha H_\alpha$ as a candidate in the projection problem. By optimality, $\mathcal J_{\rm trunc}$ is no larger than its residual. From $\pi_{\vU} \le C_\star\gauss_r$ in \textup{(H1)}, this residual is bounded by the corresponding squared gradient norm in $L^2(\gauss_r)$, multiplied by the density-ratio bound. Gaussian orthogonality and $|\alpha|_1 \geq K$ for every omitted nonzero coefficient give $\sum_{\alpha\notin\Lambda_r}|\alpha|_1\vartheta_\alpha^2 \le K^{-(2s-1)}M^2$.

For $\Lambda_{K_2,q_2}$, the bound $K\ge K_2+1$ requires all omitted positive-degree modes of degree at most $K_2$ to have zero coefficients. Within those degrees, the omitted modes are the linear $e_j$ and the interactions with $|\alpha|_0 > q_2$. If any such coefficient is nonzero, $K$ cannot exceed its total degree. In particular, some $\vartheta_{e_j} \ne 0$ implies $K = 1$, even if $\pi_{\vU}$ has zero-mean marginals.

\begin{lemma}[Exact error equation for a data-dependent regularizer]
\label{lem:unc-bias}
Let $\mA^\star$ be positive definite, let
$\mb^\star=\mA^\star\bm\theta^{r,\star}$, and let
$\widehat\mM=\widehat\mA_r+\widehat\mR$ be positive definite.
Writing
\[
  \mathbf r_N := (\widehat\mb_r-\mb^\star)
       -(\widehat\mA_r-\mA^\star)\bm\theta^{r,\star},
  \qquad
  \hvbeta_{\rm unc}:=\widehat\mM^{-1}\widehat\mb_r,
\]
we have, for every realization,
\begin{equation}
 \hvbeta_{\rm unc}-\bm\theta^{r,\star}
   =\widehat\mM^{-1}
       (\mathbf r_N-\widehat\mR\bm\theta^{r,\star}).
 \label{eq:unc-expansion}
\end{equation}
\end{lemma}
\begin{proof}
Subtract $\widehat\mM\bm\theta^{r,\star}$ from the normal equation
$\widehat\mM\hvbeta_{\rm unc}=\widehat\mb_r$.
\end{proof}

To choose the penalty, fix a deterministic regularizer $R_0\succeq0$ and set
\[
 \widetilde\beta=\bm\theta^{r,\star}
   +(\mA^\star+R_0)^{-1}
      (\mathbf r_N-R_0\bm\theta^{r,\star}).
\]
Suppose $R_0$ is diagonal in the eigenbasis of $\mA^\star$, with entries $\nu_i$.
Let $\widetilde\Sigma$ be the covariance of $\mathbf r_N$ in this basis.
Since $\EE\mathbf r_N=0$,
\begin{equation}
 \EE\|\widetilde\beta-\bm\theta^{r,\star}\|_{\mA^\star}^{2}
 =\sum_i\frac{\lambda_i[
      \nu_i^2(\theta_i^{r,\star})^2+\widetilde\Sigma_{ii}]
                }{(\lambda_i+\nu_i)^2}.
 \label{eq:bias-var-fisher}
\end{equation}
For nonzero $\theta_i^{r,\star}$, minimization gives
\begin{equation}
\nu_i^\star=\frac{\widetilde\Sigma_{ii}}{\lambda_i(\theta_i^{r,\star})^2}.
\label{eq:rho-optimal}
\end{equation}
For the actual estimator, the Gram matrix and the regularizer are both data-dependent, and \Cref{thm:ccsm-rate} bounds its error directly.

In the canonical Hermite basis, the Sobolev assumption \textup{(H2)} gives
\begin{equation}
\vartheta_\alpha^2\le M^2|\alpha|_1^{-2s}\quad(|\alpha|_1\ge1),
\qquad M^2=\|\scL_r\|_{H^s_{\gauss}}^2.
\label{eq:Hs-norm}
\end{equation}

\subsection{Per-class noise bound for the variance regularizer}
\label{app:cv-bound}

\paragraph{Setup} Let $\bm\xi(\vu) := \mathbf b(\vu) - \mathbf A(\vu)\bm\theta^{r,\star}$ be the per-observation score-matching residual, where $\mathbf A,\mathbf b$ are the Hermite integrands of Part~I. For independent exact-transform observations, the residual in \eqref{eq:unc-expansion} is $\mathbf r_N=N^{-1}\sum_n\bm\xi(\vU^{(n)})$. Its covariance is $\bm\Sigma_N=N^{-1}\Cov_{\pi_{\vU}}\bm\xi(\vU)$. In the $\mA^\star$ eigenbasis, this covariance is the matrix $\widetilde{\bm\Sigma}$ used in the deterministic-Gram calculation \eqref{eq:bias-var-fisher}.

\paragraph{Population and empirical class traces} For each $k$, the infinite-sample class trace is $\tr(\mA^{\star,(k,k)}) = \sum_{\alpha\in\Lambda^{(k)}}\sum_{j=1}^r\alpha_j\,\EE_{\pi_{\vU}}[H_{\alpha-e_j}^2(\vU)]$. Under (H1), it is bounded above by $C_\star \sum_{\alpha\in\Lambda^{(k)}}\sum_j\alpha_j = C_\star\,k\,p_k$. Write $C_+ := C_\star$ for this upper class-trace constant. A lower bound by a positive multiple of $k\,p_k$ follows if the class is non-degenerate: $\EE_{\pi_{\vU}}[H_{\alpha-e_j}^2] \ge C_- > 0$ on $\Lambda^{(k)}$. This condition holds when $\pi_{\vU}$ has full support and $\Lambda^{(k)}$ is $L^2(\pi_{\vU})$-independent. Hence
\begin{equation}
  \tr(\mA^{\star,(k,k)}) \;\asymp\; k\,p_k.
  \label{eq:class-trace-asymp}
\end{equation}
For exact i.i.d.\ samples, the Gram-matrix fluctuation is $O_p(N^{-1/2})$. The additional sample-size factor in $\widehat\sigma_k^2:=N^{-1}\tr(\widehat\mA_r^{(k,k)})$ therefore gives $\widehat\sigma_k^2-\EE\widehat\sigma_k^2=O_p(N^{-3/2})$.

\begin{corollary}[Per-mode variance comparison]
\label{cor:per-class-cov-bound}
Assume (H1) and the class-trace asymptotics \eqref{eq:class-trace-asymp}. Let $K_b$ be the Hermite hypercontractivity constant \cite{Janson1997gaussian} for total degree $\le K_{\max} := \max_{\beta\in\Lambda_r}|\beta|_1$, which is finite for a fixed multi-index set. Set $C_b := 2$ and $C_A := 2K_b^2 K_{\max}$. Then the constant
\begin{equation}
  C_\Sigma \;:=\; \frac{C_\star\,(C_b + C_A\|\bm\theta^{r,\star}\|^2)}{C_-},
  \label{eq:CSigma-def}
\end{equation}
depending only on $(C_\star, K_{\max}, K_b, \|\bm\theta^{r,\star}\|, C_-)$, satisfies, for every $\alpha\in\Lambda^{(k)}$,
\begin{equation}
  (\bm\Sigma_N)_{\alpha\alpha}
  \;\le\;
  C_\Sigma\,\widehat\sigma_k^2 \cdot \frac{k}{p_k}
  \;+\;
  O_p(N^{-3/2}).
  \label{eq:per-class-bound}
\end{equation}
The proof gives $(\bm\Sigma_N)_{\alpha\alpha}\le C_\star k^2(C_b+C_A\|\bm\theta^{r,\star}\|^2)/N$ in the canonical basis. This is the bound used in \Cref{thm:ccsm-rate}, and \eqref{eq:per-class-bound} writes it in terms of the statistic $\widehat\sigma_k^2$ of the regularizer.
\end{corollary}

\begin{proof}
\emph{(i)} The Hermite recurrence $u_j H_{\alpha-e_j} = \sqrt{\alpha_j}H_\alpha + \sqrt{\alpha_j-1}H_{\alpha-2e_j}$ reduces the $j$-th summand of $b_\alpha$ (Part~I \cite{CAS_partI}) to $\alpha_j H_\alpha$. Thus $b_\alpha = |\alpha|_1 H_\alpha = k H_\alpha$. Assumption (H1) and Hermite $L^2(\gauss_r)$-orthonormality give $\Var_{\pi_{\vU}}[b_\alpha] \le \EE_{\pi_{\vU}}[b_\alpha^2] = k^2\EE_{\pi_{\vU}}[H_\alpha^2] \le C_\star k^2$.

\emph{(ii)} Write $(\mathbf A\bm\theta^{r,\star})_\alpha = \nabla H_\alpha\cdot\nabla\scL_r^\star$. Apply Cauchy--Schwarz pointwise and Hölder in $L^4$ and $L^2$, followed by the fixed-degree Hermite hypercontractivity bound $\|f\|^2_{L^4(\gauss_r)} \le K_b\|f\|^2_{L^2(\gauss_r)}$. Using $\|\nabla H_\alpha\|^2_{L^2(\gauss_r)} = k$ and $\|\nabla\scL_r^\star\|^2_{L^2(\gauss_r)} \le K_{\max}\|\bm\theta^{r,\star}\|^2$ gives $\Var_{\pi_{\vU}}[(\mathbf A\bm\theta^{r,\star})_\alpha] \le C_\star K_b^2 k^2 K_{\max}\|\bm\theta^{r,\star}\|^2$.

\emph{(iii)} The inequality $\Var[\xi_\alpha] \le 2\Var[b_\alpha] + 2\Var[(\mathbf A\bm\theta^{r,\star})_\alpha]$ now gives $(\bm\Sigma_N)_{\alpha\alpha} = N^{-1}\Var_{\pi_{\vU}}[\xi_\alpha] \le N^{-1}C_\star k^2(C_b + C_A\|\bm\theta^{r,\star}\|^2)$. The lower bound in \eqref{eq:class-trace-asymp} is $\EE[\widehat\sigma_k^2] \ge C_-\,p_k\,k/N$. Together with the fluctuation bound above, it gives $k^2/N \le k\,\widehat\sigma_k^2/(C_-\,p_k) + O_p(N^{-3/2})$. Substitution proves \eqref{eq:per-class-bound} with the stated $C_\Sigma$.
\end{proof}

\subsection{Error bound for the Stage-2 estimator}
\label{app:excess-fisher}

\begin{theorem}[Expected Fisher error with a data-dependent regularizer]
\label{thm:ccsm-rate}
Fix a finite dictionary $\Lambda_r$ and an exact-transform basis. Assume
\textup{(H1)--(H2)}, independent observations from $\pi_{\vU}$,
$\mA^\star\succ0$, and the population score-matching identity stated above.
Let $\bm\theta^{r,\star}$ solve the population normal equations.
Fix $C_{\rm samp},C_{\rm curv}>0$ and use the regularizer of
\textup{(3.7)}. Let $\mathcal C_N$ be a measurable, nonempty,
closed convex constraint set, possibly depending on the same observations,
and let $\hvbeta$ minimize the regularized quadratic over $\mathcal C_N$.
Assume $\widehat\mM\succ0$ almost surely and
$\sup_{N\ge N_0}\EE\|\hvbeta\|^4<\infty$ for some fixed $N_0$.
The set $\mathcal C_N$ includes all imposed constraints, including centering.

Write $a:=\lambda_{\min}(\mA^\star)$,
$L:=\|\mA^\star\|_{\rm op}$,
$t_k:=\tr(\mA^{\star,(k,k)})$, and let $\mD_N$ be diagonal, with entry
\[
 \rho_k:=C_{\rm samp}t_k/N+C_{\rm curv}Lk^2
 \quad\hbox{on the degree-$k$ class.}
\]
Let $\bm\Sigma_N:=\Cov(\mathbf r_N)$ and define
\[
 \mathcal G:=\{\tfrac12\mA^\star\preceq\widehat\mA_r
                                     \preceq2\mA^\star\},
 \qquad
 d_N:=\inf_{\beta\in\mathcal C_N}
           \|\beta-\bm\theta^{r,\star}\|_{\widehat\mM}.
\]
Then, for $N\ge N_0$,
\begin{equation}
\begin{aligned}
 \EE[\mathcal J(\hvbeta)]-\mathcal J_{\rm trunc}
 &\le\mathcal E_{\rm est}+16\EE[d_N^2\mathbf1_{\mathcal G}]+O(N^{-1}),\\
 \mathcal E_{\rm est}&:=16\tr(\mD_N^{-1}\bm\Sigma_N)
 +\frac{64}{a}\|\mD_N\bm\theta^{r,\star}\|^2.
\end{aligned}
\label{eq:ccsm-rate}
\end{equation}
The constant in $O(N^{-1})$ depends on the dictionary, on $\pi_{\vU}$ and on the fourth-moment bound.

In particular, if $t_k\ge C_-kp_k$ and
$(\bm\Sigma_N)_{\alpha\alpha}\le v_0 k^2/N$
for $\alpha\in\Lambda^{(k)}$, then
\begin{equation}
 16\tr(\mD_N^{-1}\bm\Sigma_N)
 \le 16v_0\min\left\{
       \frac{\sum_{k:p_k>0}k}{C_{\rm samp}C_-},
       \frac{|\Lambda_r|}{NC_{\rm curv}L}\right\}.
 \label{eq:V-piecewise}
\end{equation}
Here $C_->0$ exists for the fixed dictionary when $\mA^\star\succ0$,
and one may take
$v_0=C_\star(2+2K_b^2K_{\max}\|\bm\theta^{r,\star}\|^2)$
from the per-mode noise calculation.
Writing $\bm\theta^{r,\star,(k)}$ for the degree-$k$ coefficient block, the other term satisfies
\begin{equation}
 \frac{64}{a}\|\mD_N\bm\theta^{r,\star}\|^2
 \le\frac{128}{a}\left[
  \frac{C_{\rm samp}^2}{N^2}
       \sum_k t_k^2\|\bm\theta^{r,\star,(k)}\|^2
 + C_{\rm curv}^2 L^2
       \sum_k k^4\|\bm\theta^{r,\star,(k)}\|^2\right].
 \label{eq:B-bound}
\end{equation}
The two bounds in \eqref{eq:V-piecewise} give the sampling scales $1/C_{\rm samp}$ and $1/(NC_{\rm curv})$ associated with the two parts of the regularizer.
\end{theorem}

\begin{proof}
Assumption \textup{(H1)} gives finite fourth moments for the polynomial
score-matching integrands. Independence of the observations then gives
$\EE\|\widehat\mA_r-\mA^\star\|_F^4=O(N^{-2})$,
$\EE\mathbf r_N=0$, and
$\bm\Sigma_N=N^{-1}\Cov[\mathbf b(\vU)
                        -\mathbf A(\vU)\bm\theta^{r,\star}]$.
On $\mathcal G$, taking principal-block traces and operator norms gives
\[
 \tfrac12\mD_N\preceq\widehat\mR\preceq2\mD_N,
 \qquad
 \mA^\star\preceq2\widehat\mM,
 \qquad
 \widehat\mM^{-1}\preceq2\mD_N^{-1},
 \quad
 \widehat\mM^{-1}\preceq(2/a)\mI.
\]
The unconstrained solution satisfies \eqref{eq:unc-expansion}. After completing
the square, the constrained minimization is the $\widehat\mM$-metric projection
of $\hvbeta_{\rm unc}$ onto $\mathcal C_N$, with solution $\hvbeta$.
Nonexpansiveness relative to any point of $\mathcal C_N$ and the triangle inequality give
\[
 \|\hvbeta-\bm\theta^{r,\star}\|_{\widehat\mM}
 \le\|\hvbeta_{\rm unc}-\bm\theta^{r,\star}\|_{\widehat\mM}+2d_N.
\]
Consequently, on $\mathcal G$,
\begin{align*}
 \mathcal J_{\rm est}(\hvbeta)
 &\le4\|\mathbf r_N-\widehat\mR\bm\theta^{r,\star}
                         \|_{\widehat\mM^{-1}}^2+16d_N^2\\
 &\le16\mathbf r_N^\top \mD_N^{-1}\mathbf r_N
       +\frac{16}{a}\|\widehat\mR\bm\theta^{r,\star}\|^2+16d_N^2\\
 &\le16\mathbf r_N^\top \mD_N^{-1}\mathbf r_N
       +\frac{64}{a}\|\mD_N\bm\theta^{r,\star}\|^2+16d_N^2.
\end{align*}
Both regularizer matrices are diagonal, so their componentwise comparison
bounds their action on the fixed vector in the last step.
Since $\mD_N$ is deterministic and $\EE\mathbf r_N=0$, after multiplication by the good-event indicator,
the expectation of the first term is bounded by $16\tr(\mD_N^{-1}\bm\Sigma_N)$.

If $\|\widehat\mA_r-\mA^\star\|_{\rm op}\le a/2$, then
$\mathcal G$ holds. Markov's inequality gives
$\Pr(\mathcal G^c)=O(N^{-2})$. Since
$\mathcal J_{\rm est}(\hvbeta)
 \le L\|\hvbeta-\bm\theta^{r,\star}\|^2$,
Cauchy--Schwarz and the coefficient fourth-moment assumption bound
$\EE[\mathcal J_{\rm est}(\hvbeta)\mathbf1_{\mathcal G^c}]$
by $O(N^{-1})$. Projection orthogonality gives
$\mathcal J(\hvbeta)=\mathcal J_{\rm trunc}
                       +\mathcal J_{\rm est}(\hvbeta)$.
This proves \eqref{eq:ccsm-rate}. For \eqref{eq:V-piecewise}, substitute
the degree-class variance bound and bound the denominator below by each summand
in $\rho_k$ in turn. Apply $(x+y)^2\le2x^2+2y^2$ to each diagonal
entry to obtain \eqref{eq:B-bound}.
\end{proof}

\begin{proof}[Proof of \Cref{thm:stage2-kl}]
The Gaussian reference scores cancel when differentiating the log density ratio, so the relative Fisher divergence is $\mathcal J(\hvbeta)$. The log-Sobolev inequality of $\widehat\pi_{\vU}$ \cite[eq.~(5.1.7)]{BakryGentilLedoux2014} gives $\DKL(\pi_{\vU}\,\|\,\widehat\pi_{\vU})\le\mathcal J(\hvbeta)/(2C_{\rm LSI}(\hvbeta))$. Orthogonality of the Fisher projection gives $\mathcal J=\mathcal J_{\rm trunc}+\mathcal J_{\rm est}$. Use the Hermite candidate from \Cref{app:bias-var} to bound the truncation term and obtain \eqref{eq:stage2-kl}. The expected bound \eqref{eq:stage2-kl-mean} follows from \Cref{thm:ccsm-rate} and \textup{(H3)}. For constants chosen on the grid, sum the bounds using the selected-pair inequality in \Cref{sec:ridge}.
\end{proof}

% =====================================================================
\section{Log-Sobolev property of the Stage-2 estimate}
\label{app:lsi-proof}

\begin{lemma}
\label{lem:lsi-pd-leading}
Let $V:\RR^r\to\RR$ be a polynomial of even degree $K \ge 4$ whose degree-$K$ homogeneous part $G$ is positive-definite, and let $Z := \int e^{-V} < \infty$. Then $\mu := Z^{-1}e^{-V}\,d\vu$ satisfies a tight logarithmic Sobolev inequality.
\end{lemma}

\begin{proof}
Positive definiteness of $G$ gives $g_0 := \min_{\|\vv\|=1}G(\vv)>0$ and hence $G(\vu) \ge g_0\|\vu\|^K$. The lower-degree terms are dominated at large $\|\vu\|$, so $V(\vu) \ge \tfrac{g_0}{2}\|\vu\|^{K}$. Thus condition (3.14.1) of \cite{CattiauxGuillinWangWu2009} holds, and $Z < \infty$. Euler's identity $\vu\cdot\nabla G = K\,G$ implies that $g_1 := \min_{\|\vv\|=1}|\nabla G(\vv)| > 0$, and homogeneity then yields $|\nabla G(\vu)| \ge g_1\|\vu\|^{K-1}$. Since $\nabla V = \nabla G + O(\|\vu\|^{K-2})$ and $\Delta V = O(\|\vu\|^{K-2})$, it follows that $|\nabla V|^2 \ge \tfrac{g_1^2}{2}\|\vu\|^{2(K-1)}$ and $|\Delta V| \le c_3\|\vu\|^{K-2}$ for large $\|\vu\|$. Set $G_{\max} := \max_{\|\vv\|=1}G(\vv)$, so that $V(\vu) \le 2G_{\max}\|\vu\|^{K}$ for large $\|\vu\|$. Because $K \ge 4$, the exponents satisfy $2(K-1) > K$. For each fixed $a_0 \in (0,1)$, choose $R$ large enough that, whenever $\|\vu\|$ exceeds this radius,
\[
(1-a_0)|\nabla V|^2 - \Delta V \ge \tfrac{(1-a_0)g_1^2}{4}\|\vu\|^{2(K-1)} \ge 2G_{\max}\|\vu\|^{K} \ge V(\vu),
\]
where the middle inequality follows from the comparison of exponents above. This is condition (3.14.2) with $\eta(u) = u$, after choosing their constant $b_0$ to cover the finite deficit on $\{\|\vu\| \le R\}$. Condition (3.14.3) follows from $V/|\nabla V|^2 = O(\|\vu\|^{2-K}) \to 0$. Both $|\nabla V|$ and $\nabla^2 V$ grow at most polynomially in $\vu$ and are therefore bounded by $e^{K'V}$ outside a ball, as required by the growth condition in their Thm.~3.17. That theorem and Rem.~3.18 of \cite{CattiauxGuillinWangWu2009} give a defective logarithmic Sobolev inequality. Conditions (3.14.1) and (3.14.2) also give a Poincar\'e inequality \cite{BakryBartheCattiauxGuillin2008}. Rothaus's lemma \cite{Rothaus1985} then gives a tight logarithmic Sobolev inequality.
\end{proof}

For \Cref{prop:lsi-holds}, take $V = \tfrac12\|\cdot\|^2 - \widehat\scL_r$ up to an additive constant, so that $\widehat\pi_{\vU} \propto e^{-V}$. When $K_{\rm eff}\ge4$, the degree-$K_{\rm eff}$ part of this potential is $-F[\hvbeta]$. It is positive-definite by the constraint $F[\hvbeta] \le -\varepsilon < 0$ on the sphere, so \Cref{lem:lsi-pd-leading} applies with $K = K_{\rm eff}$. For $K_{\rm eff} = 2$, or equivalently $K_2 \in \{2, 3\}$, the potential is strictly convex and quadratic, and the measure is Gaussian. This case follows from the Gaussian logarithmic Sobolev inequality \cite{Gross1975}.

\clearpage
\def\siamprelabel{SM}
% One reference list serves both parts: citations in the supplementary part
% print without the SM prefix that the class adds to supplement numbers.
\makeatletter\let\CAS@origcite\@cite
\def\@cite#1#2{\begingroup\def\siamprelabel{}\CAS@origcite{#1}{#2}\endgroup}\makeatother
\def\siampretitle{Supplementary Materials: }
\let\section\CASsection
\renewcommand\thesection{\siamprelabel\arabic{section}}
\crefalias{section}{section}
\crefalias{subsection}{subsection}
\setcounter{section}{0}\setcounter{equation}{0}\setcounter{figure}{0}
\setcounter{table}{0}\setcounter{theorem}{0}\setcounter{algorithm}{0}
\setcounter{page}{1}
\headers{Copula Active Subspaces}{J.~Chen and P.~J.~van Leeuwen}
\begin{center}
{\bfseries\MakeUppercase{Supplementary Materials: Copula Active Subspaces II: Error Decomposition, A Posteriori Estimation, and Sharpness of the Bounds}\par}
\vskip .075in
{\footnotesize\scshape Joshua Chen \scriptsize AND \footnotesize Peter Jan van Leeuwen}
\end{center}
\vskip .11in

\Cref{sm:rate-real} checks the $N^{-1/2}$ subspace-recovery rate of \Cref{lem:score-cov-rate} on a law whose copula score lies in $\scF_\Lambda$. \Cref{sm:disintegration-mc} explains how we compute the reference divergences. \Cref{sm:sharper-bound} compares the bound of \cite{LiCuiLiMarzoukZahm2024sharp} with the trace bound used to select the subspace, and supplies the numbers reported in \S\ref{sec:tightness} of the main manuscript. The random seeds, tuning parameters and hardware are listed in \Cref{sm:experimental-protocol}. \Cref{sm:composite-anatomy} gives the error decomposition on Examples~2 and~3.
The observation noise $\vn$ has law $\pi_{\vn}$ and is expressed in the original coordinates, before the rank transform. All other notation follows the main manuscript. References prefixed by SM (\S\ref{sm:rate-real}, \Cref{tab:sm-jhi}, \eqref{eq:sm-logcZ}) point to these supplementary materials, and references without that prefix point to the main manuscript.

% =====================================================================
\section{Numerical check of the Stage-1 rate}
\label{sm:rate-real}

We check the Stage-1 statistical rate on a law whose copula score lies in $\scF_\Lambda$, so the measured error contains no truncation error. \Cref{fig:pedagogical} illustrates the estimator in two dimensions. Part~I \cite{CAS_partI} and its supplement report results on the three example problems, including subspace recovery and convergence of the parameterized eigenspace as $(K_1,q_1)$ increases. \Cref{sm:experimental-protocol} specifies the reference bases used for those examples.

\begin{figure}[h]
\centering
\includegraphics[width=\textwidth]{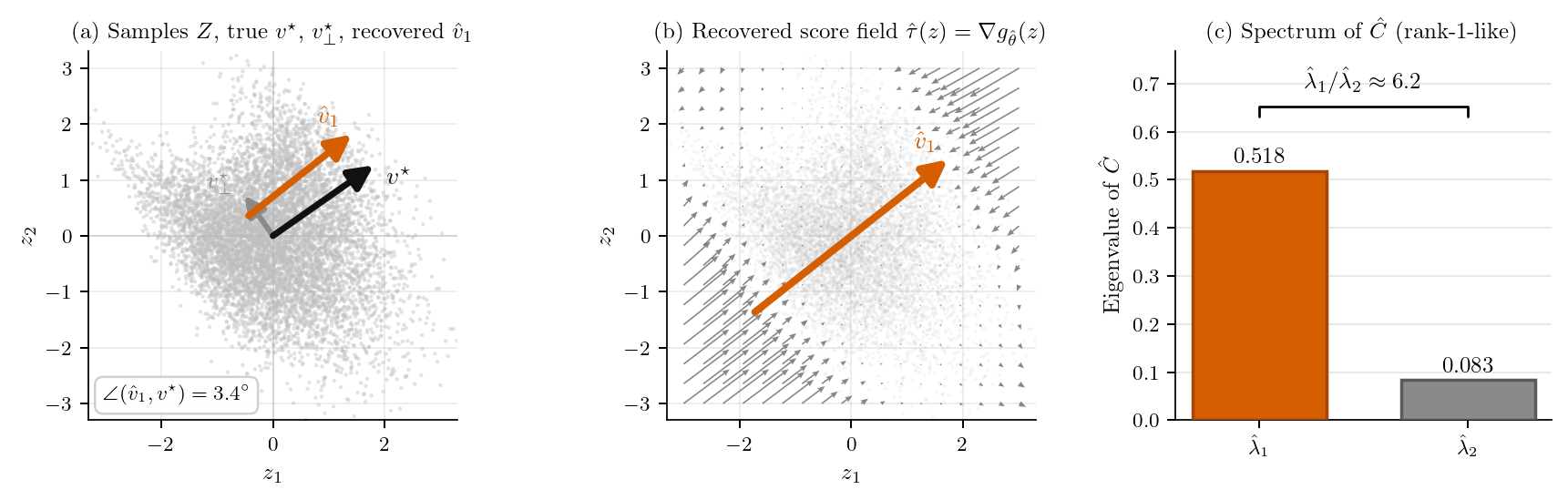}
\caption{\textsc{Cas} applied to a two-dimensional example with a rotated quadratic conditional mean, after empirical rank Gaussianization. Panel~(a) shows the samples, the construction's input direction $\vv^\star$ and the top eigenvector $\widehat{\vv}_1$ of $\hmC$. Panels~(b) and~(c) show the estimated score field and the spectrum of $\hmC$.}
\label{fig:pedagogical}
\end{figure}

\Cref{lem:score-cov-rate} of the main paper gives an $O_p(N^{-1/2})$ rate for $\|\sin\Theta(\hmV_r^{\mC}, \mV_{r,\Lambda})\|_F$ at a fixed Hermite multi-index set $\Lambda$. We check this rate on a Gaussian-copula law, for which the truncation error is exactly zero. \begin{remark}[Gaussian-copula reduction]
\label{rem:gaussian-copula}
Suppose $\pi_{\vZ} = \gauss(0,\bm\Sigma^{\vZ})$, where $\bm\Sigma^{\vZ} := \EE_{\pi_{\vZ}}[\vZ\otimes\vZ]$ is a correlation matrix. Its log-copula relative to $\gauss_d$ is the quadratic $g^\star(\vz) = -\tfrac12\vz^\top \mM\vz + \mathrm{const}$, with $\mM := (\bm\Sigma^{\vZ})^{-1} - \mI$. This is a combination of degree-$2$ Hermite polynomials and a constant, since $z_iz_j = H_{e_i+e_j}$ ($i\ne j$) and $z_i^2 = \sqrt2 H_{2e_i} + 1$.

For any $\Lambda = \Lambda_{K,q}$ with $K,q\ge2$, the space $\scF_\Lambda$ contains every symmetric linear vector field $\vz\mapsto\bm S\vz$ with $\bm S=\bm S^\top$. It therefore contains $\scS(\vz)=-\mM\vz$. Hence $\scS_\Lambda=\scS$, the coefficient vector $\theta^\star$ is supported on degree $2$, and $\mC=\mM\,\bm\Sigma^{\vZ}\mM$. Fisher consistency \cite[Thm.~2, Cor.~3]{Hyvarinen2005} identifies $\theta^\star$ as the unique infinite-sample minimizer of the strictly convex quadratic objective.

By the closed form of the solve and the strong law of large numbers, $\widehat\theta \to \theta^\star$ a.s.\ and $\hmC \to \mC$ in operator norm as $\tau_N \to 0$. At the fixed strength $\tau_N \asymp \kappa\|\hmA\|_{\rm op}$ used here, the bias term of \Cref{lem:score-cov-rate} is $O(\tau_N)$ and does not vanish, but it does not visibly change the measured rate of \Cref{fig:rate-sm} over the sample sizes used.
\end{remark}
Thus $\mV_{r,\Lambda}=\mV_r^{\mC}$, and we compute this reference eigenspace directly from the closed-form $\mC$, without Monte Carlo. The measured error includes the effects of empirical ranks and the Stage-1 ridge. We use a single correlated block of $3$ coordinates in $d=20$, with rank $r=3$, within-block off-diagonal correlation $\rho=0.7$, and $17$ independent remaining coordinates. We take $N\in\{300,500,1000,2000,4000,8000\}$, $(K_1,q_1)=(2,2)$ and $20$ repetitions at each $N$, and measure $\|\sin\Theta(\hmV_r^{\mC},\mV_r^{\mC})\|_F$. The least-squares line fitted to the mean errors in log-log coordinates has slope $-0.502$ (\Cref{fig:rate-sm}), consistent with the statistical scaling $-1/2$ in \Cref{lem:score-cov-rate}.

\begin{figure}[h]
\centering
\includegraphics[width=\textwidth]{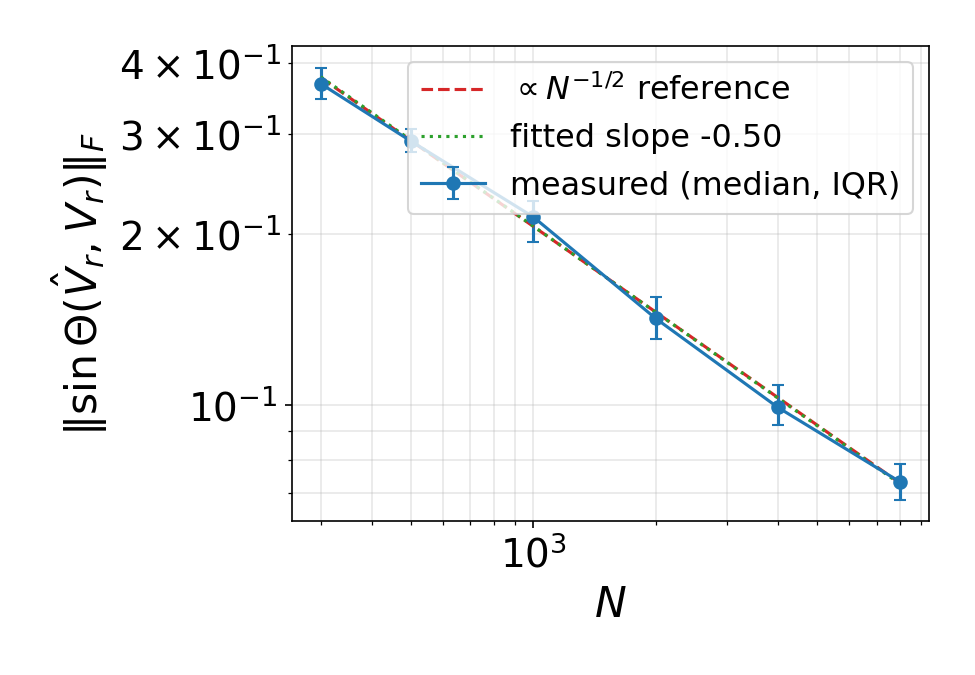}
\caption{Subspace-recovery error on a Gaussian-copula law with $d=20$ and $r=3$. A block of $3$ coordinates has within-block correlation $\rho=0.7$, and the remaining $17$ coordinates are independent. The linear copula score gives $\mV_{r,\Lambda}=\mV_r^{\mC}$ exactly. We compute this reference as the top-$r$ eigenspace of $\mC=\mM\,\bm\Sigma^{\vZ}\mM$, where $\mM=(\bm\Sigma^{\vZ})^{-1}-\mI$ (\Cref{rem:gaussian-copula}). The plot shows the median and interquartile range of $\|\sin\Theta(\hmV_r^{\mC},\mV_r^{\mC})\|_F$ over $20$ repetitions, a reference line $\propto N^{-1/2}$, and the least-squares slope fitted to the mean errors in log-log coordinates.}
\label{fig:rate-sm}
\end{figure}

% =====================================================================
\section{Reference divergences in rank-Gaussianized coordinates}
\label{sm:disintegration-mc}

We compute reference values for the Stage-1 divergence $\DKL(\pi_{\vn}\,\|\,\pi_{\vn}(\,\cdot\,;\mV_r))$ and the reduced-density divergence $\mathcal E_{\rm dens}$ using the known noise laws. These values appear in \S\ref{sec:tightness} and \S\ref{sec:exp-composite} of the main manuscript and in \Cref{sm:sharper-bound,sm:composite-anatomy}. Both calculations require pointwise values of $\pi_{\vU}(\vu):=\int\pi_{\vZ}(\mV_r\vu+\mV_\perp\vw)\,\mathrm{d}\vw$, the marginal density along $\mV_r$ after rank Gaussianization.

We apply the rank transform before evaluating this integral. Let $f_k$ and $F_k$ be the population marginal density and CDF of the mixed latent coordinate $z_k$. We represent these marginals as Gaussian mixtures over the active block. On the $b=8$ mixed coordinates, the transform is $Z_k=\Phi^{-1}(F_k(z_k))$. It is the identity on the $12$ unmixed coordinates, whose marginals are already standard Gaussian. These coordinates are independent of the mixed block, so the copula density depends only on that block. If $p_b$ is the joint density of $(z_0,\ldots,z_{b-1})$, then
\begin{equation}
  \log c^{Z}(\vZ) \;=\; \log p_b(z_0,\ldots,z_{b-1}) \;-\; \textstyle\sum_{k<b}\log f_k(z_k),
  \qquad z_k = F_k^{-1}(\Phi(Z_k)),
  \label{eq:sm-logcZ}
\end{equation}
because the standard-Gaussian factors from the changes of variables cancel. We evaluate $f_k$ and $F_k$ using $200{,}000$ mixture draws on a grid of $4{,}001$ points per mixed coordinate, and interpolate linearly between grid points. Before applying $\Phi^{-1}$, we clip $F_k$ to $[10^{-12},1-10^{-12}]$. The grid seed is $20260806$.

The integral over the complement is a Gaussian expectation of the copula density. Under the orthogonal splitting $\gauss_b = \gauss_r\otimes\gauss_{b-r}$, with $\mathbf A$ the mixed-block rows of $\mV_r$,
\begin{equation}
  c_r^{U}(\vu) \;=\; \EE_{\bm\xi \sim \gauss(\vzero,\, \mI_b - \mathbf A\mathbf A^\top)}\bigl[\,c^{Z}(\mathbf A\vu + \bm\xi)\,\bigr],
  \qquad
  \log\pi_{\vU}(\vu) = \log c_r^{U}(\vu) + \log\gauss_r(\vu),
  \label{eq:sm-crU}
\end{equation}
evaluated by a tensor Gauss--Hermite rule \cite{GolubWelsch1969} in the eigenbasis of $\mI_b - \mathbf A\mathbf A^\top$. At $\mV_r^{\mC}$ that matrix has $b-r$ unit eigenvalues and $r$ zeros, so the rule runs over $b - r = 4$ axes, with eleven nodes on each. For an estimated basis, the $r$ zero eigenvalues become positive, equal to the squared singular values of the rows of the basis on the $12$ unmixed coordinates, and are largest at small $N$. The rule then uses seven nodes on each of the $b-r$ leading axes, three on any other axis with variance exceeding $10^{-6}$, and one on the remaining axes.

The Stage-1 divergence follows as the outer average of the conditional divergence,
\begin{equation}
  \DKL\bigl(\pi_{\vn}\,\|\,\pi_{\vn}(\,\cdot\,;\mV_r)\bigr)
  \;=\;
  \EE_{\pi_{\vU}}\Bigl[\,\EE_{\bm\xi}\bigl[\rho\log\rho\bigr]\,\Bigr],
  \qquad \rho := c^{Z}(\mathbf A\vu+\bm\xi)/c_r^{U}(\vu),
  \label{eq:sm-delta-sub}
\end{equation}
over $\vu = \mV_r^\top\vZ$ with $\vZ$ drawn from $\pi_{\vZ}$.
The inner expectation in \eqref{eq:sm-delta-sub} uses the same nodes and weights as \eqref{eq:sm-crU}, so it requires no additional density evaluations. Its value is nonnegative at every $\vu$. The alternative integrand $\log c^{Z}-\log c_r^{U}$ can take either sign. We computed both averages in the same run, with outer sample size $4000$ and seed $20260806$. On every example, their means differ by less than twice the standard error of the signed average. The conditional-divergence average has a standard error smaller by a factor between $9$ and $19$.

To compute $\mathcal E_{\rm dens}$, we average $\log\pi_{\vU}(\vu_m)-\log\widehat\pi_{\vU}(\vu_m)$ over a separate sample of $\vU$, where $\widehat\pi_{\vU}$ is the Stage-2 estimate of Part~I. The reported standard errors describe the outer averages only. They do not include the errors from the marginal grids or the quadrature.

The divergences are computed after the rank transform. Before that transform, the four-dimensional subspace generated by the coupled latent coordinates has an independent standard-Gaussian complement, so retaining that subspace is exact. After the rank transform, our reference calculations give a positive projection divergence at the rank-four CAS subspace in each example.

% =====================================================================
\section{A sharper a posteriori KL divergence bound}
\label{sm:sharper-bound}

The main paper uses the Gaussian-reference KL divergence bound
\begin{equation}
  \DKL\bigl(\pi_{\vZ}\,\|\,\pi_{\vZ}(\,\cdot\,;\mV_r)\bigr) \;\le\; \tfrac12\tr((\mI-\mV_r\mV_r^\top)\mC),
  \label{eq:sm-trace-bound}
\end{equation}
with $\pi_{\vZ}(\vz;\mV_r) := c_r^{U}(\mV_r^\top\vz)\,\gauss_d(\vz)$ the reduced model \textup{(2.4)} in $\vZ$-coordinates. This bound is used in the rate analysis (\Cref{sm:rate-real}), and it gives the basis in closed form, as the top-$r$ eigenvectors of $\hmC$. When $\pi_{\vZ}$ has finite second moment and finite Fisher information in the directions orthogonal to the span of $\mV_r$, the dimensional log-Sobolev refinement of \cite{LiCuiLiMarzoukZahm2024sharp} gives a tighter upper bound, computable on the examples. Writing
\[
  \Sigma^{Z}_{\perp} := \mV_\perp^\top \EE_{\pi_{\vZ}}[\vZ\otimes \vZ]\mV_\perp,
  \qquad
  F^{Z}_{\perp} := \mV_\perp^\top \EE_{\pi_{\vZ}}[s(\vZ)\otimes s(\vZ)]\mV_\perp,
  \qquad s := \nabla \log \pi_{\vZ},
\]
for a basis $\mV_r$ and an orthogonal complement $\mV_\perp$. The two functionals in \cite{LiCuiLiMarzoukZahm2024sharp} are
\begin{align}
  J_{\mathrm{lo}}(\mV_r)
  &= \tfrac12\bigl[\,\tr(\Sigma^Z_\perp - \mI_{d-r}) - \log\det \Sigma^Z_\perp\,\bigr], \nonumber \\
  J_{\mathrm{hi}}(\mV_r)
  &= \tfrac12\bigl[\,\tr(\Sigma^Z_\perp - \mI_{d-r}) + \log\det F^Z_\perp\,\bigr],
  \label{eq:sm-jhi-jlo}
\end{align}
that satisfy, for every $\mV_r$,
\begin{equation}
  J_{\mathrm{lo}}(\mV_r)
  \;\le\;
  \DKL\!\bigl(\pi_{\vZ}\,\|\,\pi^\star_{Z, \mV_r}\bigr)
  \;\le\;
  J_{\mathrm{hi}}(\mV_r)
  \;\le\;
  \tfrac12\tr((\mI-\mV_r\mV_r^\top)\mC).
  \label{eq:sm-jhi-sandwich}
\end{equation}
The rightmost inequality follows from an algebraic identity. With $\psi(\lambda) := \lambda - 1 - \log\lambda \ge 0$,
\begin{equation}
  \tfrac12\tr((\mI-\mV_r\mV_r^\top)\mC) - J_{\mathrm{hi}}(\mV_r)
  \;=\;
  \tfrac12 \sum_{i=1}^{d-r}\psi(\lambda_i(F^Z_\perp)),
  \label{eq:sm-jhi-trace-gap}
\end{equation}
Thus the gap vanishes when $F^Z_\perp=\mI_{d-r}$. Each summand increases as its eigenvalue moves away from $1$. If $\pi_{\vZ}$ is a centered nonsingular Gaussian, the upper bound is exact for every $\mV_r$: $J_{\mathrm{hi}}(\mV_r)=\DKL(\pi_{\vZ}\,\|\,\pi^\star_{Z,\mV_r})$ \cite{LiCuiLiMarzoukZahm2024sharp}.

We minimize $J_{\mathrm{hi}}$ numerically. For the trace bound \eqref{eq:sm-trace-bound}, the minimizing subspace is given by the leading eigenvectors of $\mC$. Evaluating $J_{\mathrm{hi}}$ at a given $\mV_r$ requires the moments $\Sigma^Z_\perp$ and $F^Z_\perp$. We estimate both moments from samples of the rank-Gaussianized $\vZ$ and the copula score of \Cref{sm:disintegration-mc}, so these estimates do not involve Hermite score matching. The marginal estimation in \Cref{sm:disintegration-mc} remains a source of numerical error.

\paragraph{Bounds at $\mV_r^{\mC}$ and after descent} \Cref{tab:sm-jhi} reports, for each example, the trace bound from the $10^6$-sample reference for $\mC$, the sharper bound $J_{\mathrm{hi}}$ and the reference divergence. We compute $J_{\mathrm{hi}}$ by subtracting the gap \eqref{eq:sm-jhi-trace-gap} from the trace bound. The gap uses the eigenvalues of $F^Z_\perp$ evaluated at the same basis $\mV_r^{\mC}$, from a separate moment estimate based on $2\cdot10^6$ samples.

We also minimize this form of $J_{\mathrm{hi}}$ by Riemannian gradient descent \cite{AbsilMahonySepulchre2008}. For the descent, both $\mC$ and $F^Z_\perp$ are estimated from the $2\cdot10^6$ samples, and the initial basis consists of the leading $r$ eigenvectors of that estimate of $\mC$. Subtracting the nonnegative gap keeps the computed $J_{\mathrm{hi}}$ below the computed trace bound. This evaluation is equivalent to \eqref{eq:sm-jhi-jlo} with $\EE_{\pi_{\vZ}}[\vZ\otimes\vZ]$ replaced by $\mC-\mH+2\mI$, using the population identity $\mC=\mM+\mH-2\mI$ of \Cref{sm:experimental-protocol}.

\begin{table}[h]
\centering\small
\caption{Bounds and reference divergences at $\mV_r^{\mC}$ on Examples~1--3, in nats. We compute $J_{\mathrm{hi}}$ by subtracting the gap $\tfrac12\sum_i\psi(\lambda_i(F^Z_\perp))$ from the trace bound. The divergence is reported with its Monte Carlo standard error. The last column gives the distance moved by gradient descent on $J_{\mathrm{hi}}$.}
\label{tab:sm-jhi}
\begin{tabular}{l|ccc|c|c}
\toprule
Example & trace bound & $J_{\mathrm{hi}}$ & gap & $\DKL$ & $\|\sin\Theta\|_F$ \\
\midrule
Example~1 & $0.0181$ & $0.0181$ & $4.9\cdot10^{-5}$ & $0.0095{\pm}0.0001$ & $3.2\cdot10^{-5}$ \\
Example~2 & $0.0451$ & $0.0449$ & $2.0\cdot10^{-4}$ & $0.0275{\pm}0.0004$ & $3.2\cdot10^{-4}$ \\
Example~3 & $0.0182$ & $0.0181$ & $7.3\cdot10^{-5}$ & $0.0092{\pm}0.0002$ & $3.7\cdot10^{-5}$ \\
\bottomrule
\end{tabular}
\end{table}

\paragraph{Comparison}

The gap is below $0.5\%$ of the trace bound on each example, and the estimated eigenvalues of $F^Z_\perp$ lie in $[0.995,1.024]$. For a basis supported on the mixed block, $F^Z_\perp$ has twelve unit eigenvalues corresponding to the twelve independent standard-Gaussian coordinates.

The trace bound exceeds the divergence by factors of $1.91$, $1.64$ and $1.98$ on Examples~1--3. The corresponding factors for $J_{\mathrm{hi}}$ are $1.90$, $1.63$ and $1.97$. The sharper bound therefore removes little of the overestimate on these examples.

The optimization stops when the line search finds no decrease. It moves the basis by $\|\sin\Theta\|_F\le0.0004$ and reduces $J_{\mathrm{hi}}$ by less than $2\cdot10^{-6}$ of its value. For comparison, the leading eigenspaces estimated from $10^6$ and $2\cdot10^6$ samples differ by $\|\sin\Theta\|_F=0.0014$, $0.0019$ and $0.0013$ on Examples~1--3. Each of these differences is larger than the movement during descent.

\section{Experimental Protocol}
\label{sm:experimental-protocol}

All experiments are deterministic under the fixed random seeds listed below. Rank Gaussianization assigns average ranks to ties and uses the $k/(N+1)$ plotting-position convention. The randomized range finder \cite{HalkoMartinssonTropp2011} uses oversampling $p=10$ and $q_{\text{pwr}}=2$ power-iteration steps. The Python implementation uses NumPy 2.4 and SciPy 1.17, as pinned in the release. The Reproducibility section of the main manuscript gives the code repository.

\paragraph{Example problems and reference bases (\Cref{sm:rate-real})} The three examples have $d=20$, with dependence introduced in four latent coordinates. As specified in Part~I \cite{CAS_partI}, an orthogonal matrix $\mQ$ mixes the first $b=8$ coordinates and leaves the remaining $12$ unchanged. The observation map is $\operatorname{logit}\circ\Phi$, applied componentwise. The latent coordinates have identity covariance. After the rank transform, each coordinate of $\vZ$ is standard Gaussian, and the coordinates of the mixed block are dependent. The matrix $\mC$ has rank at most $b=8$. We use $r=4$, the number of coupled latent coordinates.

The numerical reference for $\mV_r^{\mC}$ uses $N_{\rm ref}=10^6$ oracle-score samples with random seed $20260715$ and mixture seed $20260716$. The Stage-2 reference values use a different basis, $\mV_{r,\Lambda,\tau}$: the leading $r$ eigenvectors of the Stage-1 estimate obtained from $N_{\rm ref}=10^6$ samples with $(K_1,q_1)=(4,2)$ and ridge $\kappa=10^{-4}$.

The parameters of the separate Gaussian-copula rate experiment are given in \Cref{sm:rate-real}. It uses data seed $500+20i+j$, where $i=0,\ldots,5$ indexes the sample sizes and $j=0,\ldots,19$ indexes the repetitions. Its Stage-1 estimate uses seed $500+j$.

\paragraph{Sharper-bound evaluation (\Cref{sm:sharper-bound})} We estimate the moments from $2\cdot10^6$ samples of $\vZ$ with seed $0$ and mixture seed $1$. The trace bound in \Cref{tab:sm-jhi} uses the $10^6$-sample reference described above. All these moments are evaluated after rank Gaussianization. In the latent coordinates $\vz=\mQ\vW$, the score relative to the standard Gaussian is $\mQ(\nabla\log\pi_{\vW}+\vW)$. This score lies in a four-dimensional subspace, so the corresponding rank-$4$ trace bound would be zero.

Integration by parts gives $\mC=\mM+\mH-2\mI$, with $\mM:=\EE[\vZ\vZ^\top]$ and $\mH:=\EE[s(\vZ)s(\vZ)^\top]$. The gap \eqref{eq:sm-jhi-trace-gap} is a sum of nonnegative terms $\psi(\lambda_i)$, and \Cref{tab:sm-jhi} computes $J_{\mathrm{hi}}$ in that form.

\paragraph{Stage-1 terms (\S\ref{sec:exp-aposteriori} of the main manuscript)} For $\scS'$ in the a priori bound \eqref{eq:finite-sample-kl}, we use the Stage-1 estimate fitted to $N_{\rm ref}=10^6$ residuals with $(K_1,q_1)=(4,2)$. We evaluate $\|\scS-\scS'\|^2$ on $2\cdot10^5$ draws of $\vZ$, using the copula score from \Cref{sm:disintegration-mc}. Its average estimates $\|\scS-\scS'\|_{L^2(\pi_{\vZ})}^2$, and the rank and statistical terms use the second-moment matrix of $\scS'$ on the same draws.

We use these draws also for $e_N$ and $D_N$ in \eqref{eq:aposteriori-corrected}. To evaluate $D_N$, we map the corresponding residuals through the rank transform fitted on the first half of the sample. For the saturation constant $\beta$, we compute the unpenalized projections of the copula score onto the two nested gradient spaces, using $2\cdot10^5$ draws of $\vZ$.

The reference fit uses seed $20260701$. The fit with sample size $N$ and seed index $s$ uses seed $20260701+1009s+7N$. The $2\cdot10^5$ evaluation draws use seed $20261700$ and the marginal tables of \Cref{sm:disintegration-mc}. The draws used to compute $\beta$ have seed $20260715$ and mixture seed $20260716$.

\paragraph{End-to-end comparison (\S\ref{sec:exp-endtoend} of the main manuscript)} We fit the model with seed index $s$ at basis $\mV_{r,\Lambda,\tau}$, using residuals drawn with seed $7000+s$. The $M=2\cdot10^5$ held-out residuals use seed $1277000+1000j+s$, where $j=0,\ldots,4$ indexes $N$ in $(500,2500,12500,50000,200000)$. The reference values are computed as described in \Cref{sm:composite-anatomy}.

\paragraph{Stage-2 normalizer (\S\ref{sec:aposteriori} of the main manuscript)} We approximate $\widehat Z_r=\EE_{\gauss_r}[\exp(\widehat\scL_r)]$ with $M_{\rm norm}=8192$ scrambled Sobol' points \cite{Sobol1967,Owen1995}, mapped componentwise through $\Phi^{-1}$. The scramble uses the model seed. This normalizer is needed to evaluate $\log\widehat\pi_{\vU}$ and $\log\widehat\pi_{\vn}$, but not to fit the coefficients.

\paragraph{Computation} The Stage-1 and Stage-2 estimates are computed in one pass over the sample, in batches. The decompositions of the different estimated models, and the rate check, run in parallel. The computation over the sequence of larger multi-index sets estimates the peak memory for each set before starting it and skips any set that would exceed $35$~GB. No set reported here was skipped, the largest being $\Lambda_{6,3}$ with $25{,}770$ multi-indices.

% =====================================================================
\section{The three terms on Examples 2 and 3}
\label{sm:composite-anatomy}
\Cref{fig:composite-ex2,fig:composite-ex3} give the results for Examples~2 and~3 under the protocol of \Cref{fig:composite-main}. We plot the sum $\mathcal E_{\rm tot}:=\mathcal E_{\rm marg}+\mathcal E_{\rm proj}+\mathcal E_{\rm dens}$. For model seed $s$, the reference full and marginal divergences are averages over the same evaluation sample, with seed $7000+s+10^6$. The reference $\mathcal E_{\rm dens}$ uses a separate sample of $\vU$ with seed derived from $7000+s+2\cdot10^6$. We compute $\mathcal E_{\rm proj}$ by the quadrature in \Cref{sm:disintegration-mc}. The reference marginal densities are kernel estimates fitted to $5\cdot10^6$ draws.

The plotted components differ from the exact decomposition in \Cref{prop:exact-decomp} in two respects. The reference marginal term omits $R_N$ and $\log\widehat Z_{\mathrm{comp}}$. Also, $\mathcal E_{\rm dens}$ is evaluated at $\mV_{r,\Lambda,\tau}$, whereas $\mathcal E_{\rm proj}$ is evaluated at the estimated basis. The difference between the plotted sum and the reference total therefore includes these two omitted corrections, the change in $\mathcal E_{\rm dens}$ between the two bases, and numerical errors in the reference values.

\begin{figure}[h]
\centering
\includegraphics[width=\textwidth]{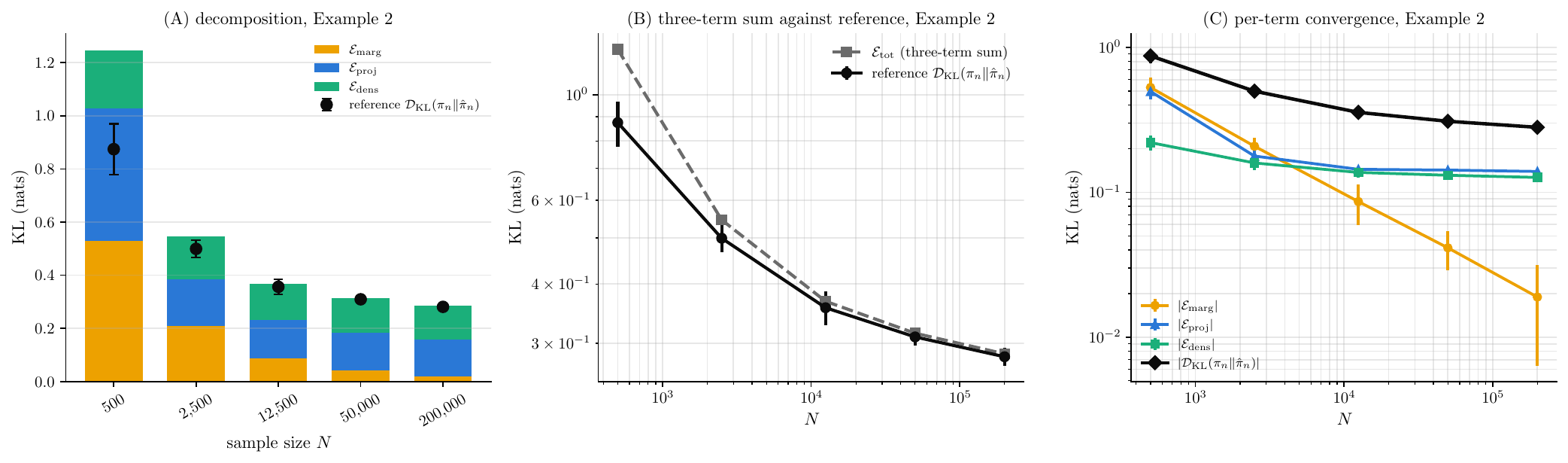}
\caption{Reference components of the error on Example~2, computed as described in \Cref{sm:composite-anatomy}. The marginal term $\mathcal E_{\rm marg}$ is largest for $N\le2500$ and decreases as $N$ grows. The subspace term $\mathcal E_{\rm proj}$ levels off and is largest for $N\geq12{,}500$. The reduced-density term $\mathcal E_{\rm dens}$ levels off slightly below it. For $N\ge12{,}500$, the plotted sum $\mathcal E_{\rm tot}$ nearly coincides with the reference total $\DKL$.}
\label{fig:composite-ex2}
\end{figure}

\begin{figure}[h]
\centering
\includegraphics[width=\textwidth]{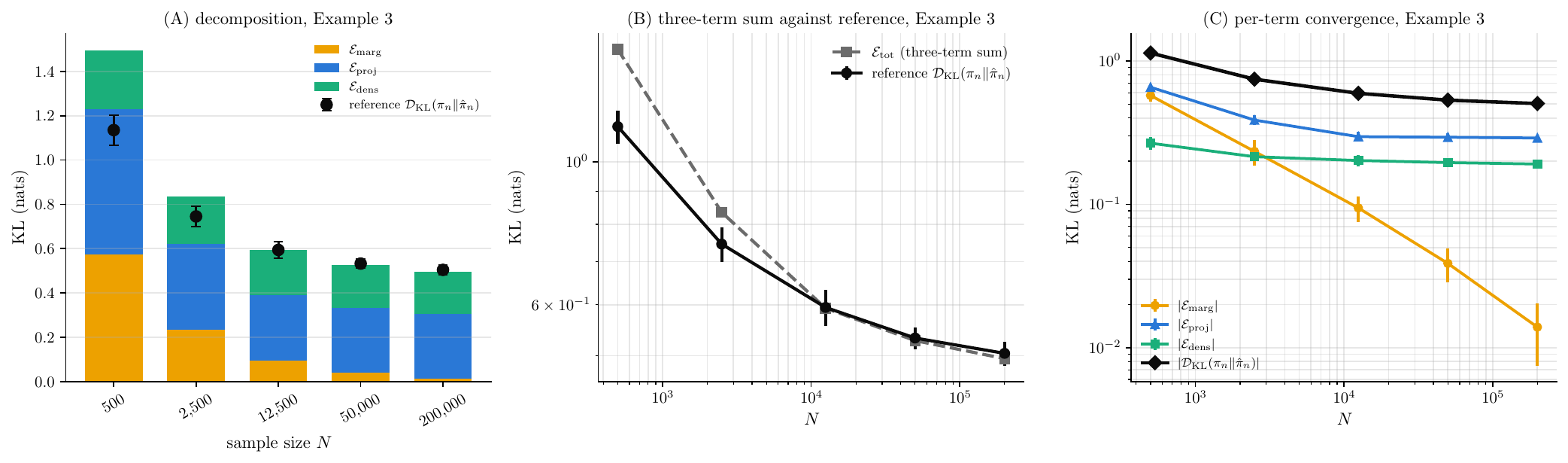}
\caption{Reference components of the error on Example~3, computed as described in \Cref{sm:composite-anatomy}. The subspace term $\mathcal E_{\rm proj}$ is largest at every $N$ and levels off as $N$ grows. The marginal term $\mathcal E_{\rm marg}$ decreases and falls below $\mathcal E_{\rm dens}$ between $N=2500$ and $N=12{,}500$. As on Example~2, the plotted sum $\mathcal E_{\rm tot}$ nearly coincides with the reference total $\DKL$ for $N\ge12{,}500$.}
\label{fig:composite-ex3}
\end{figure}

\clearpage
\def\siamprelabel{}
\end{document}